\documentclass[draftcls,onecolumn,12pt]{IEEEtran}

\usepackage{textcomp,amssymb,amsmath}
\usepackage{makeidx,epsfig}
\usepackage{graphicx,subfigure}
\usepackage{color}
\usepackage{float}
\usepackage{mathrsfs}
\usepackage{dsfont}
\usepackage{theorem}
\usepackage{algorithm,algorithmic}
\usepackage{cite}

\def \ee{\mathbb{E}}

\def \sd{\mathbf{s}_d}
\def \yd{\mathbf{y}_d}
\def \zd{\mathbf{z}_d}
\def \yr{\mathbf{Y}_r}
\def \vr{\mathbf{V}_r}
\def \sr{\mathbf{s}_r}

\def \bh {\mathbf{h}}
\def \bg {\mathbf{g}}
\def \bs {\mathbf{s}}
\def \by {\mathbf{y}}
\def \bv {\mathbf{v}}

\def \bz {\mathbf{z}}
\def \bA {\mathbf{A}}
\def \bI {\mathbf{I}}
\def \bN {\mathbf{N}}

\def \pr{P_r}
\def \pd{P_d}

\def \td{T_d}
\def \plim{\lim\limits_{P\rightarrow \infty}}

\def \dRlGen{ \Delta R^{(l)}}
\def \dRuGen{ \Delta R^{(u)}}

\def \RANdata {R_{\rm conv}}
\def \tRANdata {\tilde{R}_{\rm conv}}
\def \dRlANdata { \Delta R^{(l)}_{\rm conv}}
\def \dRuANdata { \Delta R^{(u)}_{\rm conv}}

\def \RANboth {R_{\rm DCE}}
\def \tRANboth {\tilde{R}_{\rm DCE}}
\def \dRlANboth { \Delta R^{(l)}_{\rm DCE}}
\def \dRuANboth { \Delta R^{(u)}_{\rm DCE}}

\def \sh{\sigma_{h}^2}
\def \sg{\sigma_{g}^2}
\def \sdh{\sigma^2_{\Delta h}}
\def \sdg{\sigma^2_{\Delta g}}
\def \sdhr{\sigma_{\Delta h_r}^2}
\def \shh{\sigma_{h}^4}
\def \sgg{\sigma_{g}^4}
\def \sdhh{\sigma^4_{\Delta h}}
\def \sdgg{\sigma^4_{\Delta g}}

\def \hh{\|\mathbf{\hat h}\|^2}
\def \gg{\frac{\mathbf{\hat g^{\dag} \hat h \hat h^{\dag} \hat g}}{\hh}}

\def \ng{\|\mathbf{N}_{\hat \bh}\mathbf{\hat g}\|^2}
\def \ngg{\|\mathbf{N}_{\hat \bh} \mathbf{\overline{g}}\|^2}

\def \dhr{\mathbf{\Delta h}_r}

\def \effh {\|\mathbf{\overline h}\|^2}
\def \effg {\frac{\mathbf{\overline{g}^{\dag}\overline{h} \overline{h}^{\dag}\overline{g}}}{\effh}}

\def\calP{\mathcal{P}}
\def\calQ{\mathcal{Q}}

\def\tr{{\rm tr}}

\newtheorem{Theo}{Theorem}
\newtheorem{Prop}{Proposition}
\newtheorem{Lem}{Lemma}
\newtheorem{Cor}{Corollary}

\begin{document}

\title{On the Role of Artificial Noise in Training and\\[-.2cm] Data Transmission for Secret Communications}

\author{Ta-Yuan Liu, Shih-Chun Lin, and Y.-W. Peter Hong

\thanks{T.-Y. Liu and Y.-W. P. Hong (emails: {\tt tyliu@erdos.ee.nthu.edu.tw} and {\tt ywhong@ee.nthu.edu.tw}) are with the Institute of Communications Engineering, National Tsing Hua University, Hsinchu, Taiwan 30013.
S.-C. Lin (email: {\tt sclin@mail.ntust.edu.tw}) is with the Department of Electronic and Computer Engineering, National Taiwan University of Science and Technology, Taipei, Taiwan 10607.}\thanks{This work was presented in part at IEEE International Conference on Communications (ICC) 2012, and was supported in part by the Ministry of Science and Technology, Taiwan, under grants 104-3115-E-007-003 and 102-2221-E-007-016-MY3.}}

\maketitle

\vspace{-.8cm}

\begin{abstract}\vspace{-.2cm}
This work considers the joint design of training and data transmission in physical-layer secret communication systems, and examines the role of artificial noise (AN) in both of these phases. In particular, AN in the training phase is used to prevent the eavesdropper from obtaining accurate channel state information (CSI) whereas AN in the data transmission phase can be used to mask the transmission of the confidential message. By considering AN-assisted training and secrecy beamforming schemes, we first derive bounds on the achievable secrecy rate and obtain a closed-form approximation that is asymptotically tight at high SNR. Then, by maximizing the approximate achievable secrecy rate, the optimal power allocation between signal and AN in both training and data transmission phases is obtained for both conventional and AN-assisted training based schemes. We show that the use of AN is necessary to achieve a high secrecy rate at high SNR, and its use in the training phase can be more efficient than that in the data transmission phase when the coherence time is large. However, at low SNR, the use of AN provides no advantage since CSI is difficult to obtain in this case. Numerical results are presented to verify our theoretical claims.
\end{abstract}

\begin{keywords}\vspace{-.2cm}
Secrecy, wiretap channel, channel estimation, artificial noise, power allocation.
\end{keywords}

\section{Introduction} \label{sec.into}

Information-theoretic secrecy has received renewed interest in recent years, especially in the context of wireless communications, due to the broadcast nature of the wireless medium and the increasing amount of confidential data that is being transmitted over the air. Most of these studies stem from the seminal works by Wyner in \cite{ADWyner1975} and by Csiszar and Korner in \cite{ICsiszar1978}, where the so-called secrecy capacity was characterized for degraded and nondegraded discrete memoryless wiretap channels (i.e., channels consisting of a source, a destination, and a passive eavesdropper), respectively. The notion of secrecy capacity was introduced in these works as
the maximum achievable secrecy rate between
the source and the destination
subject to a
constraint on the information attainable by the eavesdropper. These issues were also examined
for Gaussian channels by Leung-Yan-Cheong and Hellman in \cite{SLeung-Yan-Cheong1978}, where Gaussian signalling was shown to be optimal. 
These works show that the secrecy capacity of a wiretap channel increases with the difference between the channel quality at the destination and that at the eavesdropper.

In recent years, studies of the wiretap channel have also been extended to multi-antenna wireless systems, e.g., in \cite{AKhisti2010_MISOME1,AKhisti2010_MIMOME2,SShafiee_Ulukus2009_221,FOggier_Hassibi2011,TLiu_Shamai2009}, where the achievable secrecy rates were examined under different channel assumptions and techniques were proposed to best utilize the available spatial degrees of freedom. In particular, the work in \cite{AKhisti2010_MISOME1} examined the secrecy capacity of a multiple-input single-output (MISO) wiretap channel and showed that transmit beamforming with Gaussian signalling is optimal. However, perfect knowledge of both the main and the eavesdropper channel state information (CSI) was required at the source in order to determine the optimal beamformer. In \cite{AKhisti2010_MIMOME2,SShafiee_Ulukus2009_221,FOggier_Hassibi2011,TLiu_Shamai2009}, more general results were obtained for cases with multiple antennas at the destination. Precoding techniques were proposed as a generalization of the beamforming scheme in \cite{AKhisti2010_MISOME1} to higher dimensions and, thus, perfect CSI of all links was also required to derive the optimal precoder.
On the other hand, when the eavesdropper CSI is unavailable, which is often the case in practice,
the secrecy capacity and its corresponding optimal
transmission scheme are both unknown. However, an artificial noise (AN) assisted secrecy beamforming scheme, where data is beamformed towards the destination and AN is placed in the null space of the main channel direction to jam the eavesdropper's reception,
is often adopted and was in fact shown to be asymptotically optimal in \cite{AKhisti2010_MISOME1}.
Even though knowledge of the eavesdropper channel is not required in this transmission scheme, perfect knowledge of the main channel CSI is still needed, which
can also be unrealistic
due to the presence of noise in the channel estimation.

In practice, CSI is typically obtained through training and channel estimation at the destination. In conventional systems (without secrecy constraints), training signal designs
have been studied in the literature for both
single-user \cite{MBiguesh_Gershman2006,TFWong_Park2004}
and multiuser systems \cite{JWChen_Ng2008}.
In these cases, training
is often done by having the source transmit a pilot signal to enable channel estimation at the destination (and CSI at the source is obtained by having the destination feedback its channel estimate to the source).  However, this approach may not be favorable for systems with secrecy considerations since the emission of pilot signals by the source also enables channel estimation at the eavesdropper
(and in this way enhances its ability to intercept the source's message). More recently, a secrecy enhancing training scheme, called the discriminatory channel estimation (DCE) scheme, was proposed in \cite{Chang_Chi2010_DCE,CWHuang_Hong2013_TwoWay}, where AN is super-imposed on top of the pilot signal in the training phase to disrupt the channel estimation at the eavesdropper.
These works showed that DCE can indeed enhance the difference between the channel estimation qualities at the destination and the eavesdropper in the training phase (before the actual data is transmitted), but did not discuss its impact on the achievable secrecy rate in the data transmission phase.

The main objective of this work is to examine the impact of both conventional and DCE-type training on the achievable secrecy rate of AN-assisted secrecy beamforming schemes. Different from previous works in the literature that focus on either training or data transmission,
we consider the joint design and examine the role of AN in both of these phases. In this work, the two-way DCE scheme proposed in \cite{CWHuang_Hong2013_TwoWay}  is employed in the training phase to prevent CSI leakage to the eavesdropper, and the AN-assisted secrecy beamforming scheme is used in the data transmission phase to mask the transmission of the confidential message.
We first derive bounds on the achievable secrecy rate of these schemes, which are shown to be asymptotically tight as the transmit power increases, and utilize them to obtain closed-form approximations of the achievable secrecy rate.
Then, based on the approximate secrecy rate expressions, optimal power allocation policies for the pilot signal, the data signal, and AN in both phases are obtained for systems employing conventional and AN-assisted training schemes, respectively. We show that the use of AN (in either training or data transmission) is often necessary to achieve a significantly higher secrecy rate at high SNR, and that its use in training can be more efficient than that in data transmission when the coherence time is long. However, in the low SNR regime, the use of AN provides no advantage in either training or data transmission. In fact, allocating resources for training can be strictly
suboptimal in this regime since it is difficult to obtain useful CSI when power is scarce. Numerical results are provided to verify our theoretical claims.

The joint design of training and data transmission have been investigated for conventional MIMO
point-to-point and multiuser scenarios (without secrecy constraints) in
\cite{BHassibi_Hochwald2003_Training} and \cite{GCaire2010}, respectively.
However, these issues have not been discussed before for physical layer secret communications, where
finding a reasonable approximation for the achievable secrecy rate under channel estimation errors, and coping with the non-Gaussianity caused by the combination of AN and channel estimation errors can be challenging. The impact of imperfect CSI due to channel estimation errors and limited feedback on the achievable secrecy rate have been examined in \cite{ZRezki_Khisti2014,Zhou_McKay_2010} and \cite{SCLin_Chi2011_Quantized_Feedback}, respectively. However, these works focus on the achievable secrecy rate for given estimation error statistics without consideration on how training should be performed and how it can impact the error statistics. Moreover, CSI at the eavesdropper is often assumed to be perfect in these works to avoid the need to analyze the impact of channel estimation error at the eavesdropper. A preliminary study of our work was presented in \cite{TYLiu_ICC} for the case of conventional training. The current work further considers the case of AN-assisted training, provides rigorous proofs of the theoretical claims, and examines the low SNR case.

The remainder of this paper is organized as follows. In Section \ref{sec.syst}, the system model and the training-based transmission scheme are introduced. In Section \ref{sec.secrecy rate}, upper and lower bounds of the achievable secrecy rate under channel estimation error are obtained. In Sections \ref{sec.no AN training} and \ref{Sec.DCE}, closed-form secrecy rate expressions and optimal power allocation policies are derived for cases with conventional and DCE training, respectively. The analysis of the secrecy rate with training-based transmission scheme in the low SNR regime is discussed in Section \ref{Sec.low power regime}. Finally, numerical results are provided in Section \ref{sec.simulation}, and a conclusion is given in Section \ref{sec.conclusion}.

\section{System Model}\label{sec.syst}


\begin{figure}[t]
\begin{center}
  \includegraphics[width=11cm]{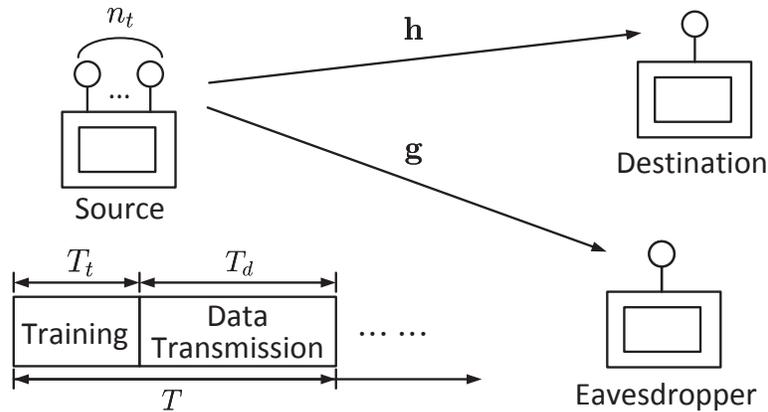}\\
    \vspace{-0.2in}
  \caption{Training-based secret transmission scheme that consists of a training phase and a data transmission phase.
   }\label{fig:system}
\end{center}
\vspace{-0.3in}
\end{figure}

Let us consider a wireless secret communication system that consists of a source, a destination, and an eavesdropper. The source is assumed to have $n_t$ antennas whereas both the destination and the eavesdropper are assumed to have only a single antenna each. The main and eavesdropper channels (i.e., the channel from the source to the destination and to the eavesdropper, respectively) can be described by the vectors $\mathbf{h}=[h_1,\ldots, h_{n_t}]^T$ and $\mathbf{g}=[g_1,\ldots, g_{n_t}]^T$, respectively, where the entries are assumed to be independent and identically distributed (i.i.d.) complex Gaussian random variables with mean $0$
variances $\sigma_h^2$ and $\sigma_g^2$, respectively (i.e., ${\cal CN}(0,\sigma_h^2)$ and ${\cal CN}(0,\sigma_g^2)$).
We consider a block fading scenario where the channel vectors remain constant over a coherence interval of duration $T$, but vary independently from block to block. By adopting a training-based transmission scheme, each coherence interval is divided into a training phase with duration $T_t$ and a data transmission phase with duration $T_d$, as illustrated in Fig. \ref{fig:system}. In the training phase, pilot signals are emitted by the source (and/or the destination) to enable channel estimation at the destination; and, in the data transmission phase, confidential messages are transmitted utilizing the estimated channel obtained in the previous phase.
Following methods proposed in \cite{Chang_Chi2010_DCE,CWHuang_Hong2013_TwoWay} for training and in \cite{SGoel_Negi2008} for data transmission, AN is utilized in the respective phases to degrade the reception at the eavesdropper. Our goal is thus to determine the optimal resource allocation between signal and AN, and examine the role of AN in these two phases.



\subsection{Training Phase - AN-Assisted Training}\label{sec.syst.training}

In conventional point-to-point communication systems, training is typically performed by having the source emit pilot signals to enable channel estimation at the destination. Most works in the literature on physical layer secrecy, e.g., \cite{AKhisti2010_MISOME1,SGoel_Negi2008,PHLin_SCLin2013,TYLiu_Mag}, inherit such an assumption and, thus, assume that
the eavesdropper can also benefit from the pilot transmission and can obtain a channel estimate that is no worse than the destination.
Interestingly, it has been shown more recently in \cite{Chang_Chi2010_DCE,CWHuang_Hong2013_TwoWay} that secrecy can be further enhanced by embedding AN in the pilot signal to degrade the channel estimation performance at the eavesdropper. By doing so, the difference between the effective channel qualities experienced by the destination and the eavesdropper can be enhanced and, thus, a higher secrecy rate can be achieved. Here, we consider in particular the two-way discriminatory channel estimation (DCE) scheme proposed in \cite{CWHuang_Hong2013_TwoWay}.
In the DCE scheme, training is performed in two stages, i.e., the reverse and the forward training stages.
In the reverse training stage, a pure pilot signal is sent in the reverse direction by the destination to enable channel estimation at the source; in the forward training stage, a pilot signal masked by AN is emitted by the source to facilitate channel estimation at the destination while preventing reliable channel estimation at the eavesdropper. Here, the channel is assumed to be reciprocal, that is, the reverse channel can be represented as the transpose of the forward channel vector, i.e., $\bh^t$. Therefore, estimation of the reverse channel provides the source with information about the forward channel.
Note that DCE can also be used in non-reciprocal channels, as shown in \cite{CWHuang_Hong2013_TwoWay}, but is not considered here for simplicity.


Let $T_r$ and $T_f$ be the length of the reverse and the forward training stages, respectively, where $T_r+T_f=T_t$.
In the reverse training stage, the pilot signal $\sr \in \mathcal{C}^{T_{r} \times 1}$ with $\sr^{\dag}\sr=T_r$ is first emitted by the destination and the received signal at the source can be written as
\begin{align}\label{reverse training}
\yr=\sqrt{\pr}\,\sr \mathbf{h}^{t}+\vr
\end{align}
where
$\pr$ is the power of the pilot signal in the reverse training stage,
$\mathbf{h}^{t}$ is the channel vector from the destination to the source, and $\vr\in\mathcal{C}^{T_r\times n_t}$ is the additive white Gaussian noise (AWGN) matrix with entries that are i.i.d.
$\mathcal{CN}(0, \sigma^2)$.
Following the procedures given in \cite{CWHuang_Hong2013_TwoWay}, the source first computes the minimum mean square error (MMSE) estimate of the channel based on the knowledge of $\sr$. The channel estimate at the source is denoted by $\mathbf{\widetilde{h}}$ and the channel estimation error is $\dhr=\mathbf{h}-\mathbf{\widetilde{h}}$. The variance of each entry of $\dhr$ can be written as
\begin{align}\label{var error h reverse training}
\sdhr=\left(\frac{1}{\sh}+\frac{P_rT_r}{\sigma^2}\right)^{-1}.
\end{align}

Then, in the forward training stage, the source emits a training signal with AN placed in the null space of the estimated forward channel, i.e., $\mathbf{\widetilde{h}}$. The signal transmitted in the forward training stage is given by
\begin{equation}\label{eq_forward_training_signal}
\mathbf{X}_f=\sqrt{P_f} \mathbf{S}_f +\mathbf{A}_f\mathbf{N}_{\widetilde \bh},
\end{equation}
where $\mathbf{S}_f\in \mathcal{C}^{T_f \times n_t}$ is the pilot signal in the forward training stage with $\mathbf{S}^{\dag}_f\mathbf{S}_f=\frac{T_f}{n_t}\mathbf{I}$, $P_f$ is the power of the pilot signal in the forward training stage,
$\mathbf{N}_{\widetilde \bh}\in \mathcal{C}^{(n_t-1)
\times n_t}$ is a
matrix whose rows span the null space of $\mathbf{\widetilde h}$ and
satisfies $\mathbf{N}_{\widetilde \bh}\mathbf{N}^{\dag}_{ \widetilde
\bh}=\mathbf{I}_{n_t-1}$, and $\mathbf{A}_f \in \mathcal{C}^{T_f
\times (n_t-1)}$ is the AN with entries that are
i.i.d. $\mathcal{CN}(0,\frac{P_{f_a}}{n_t-1})$.
Hence, the total
AN power in the forward training stage is $P_{f_a}$.
The signals received at the destination and the eavesdropper can then be written respectively as
\begin{align} \label{eq forward training_rx_LR}
\mathbf{y}_f&=\mathbf{X}_f\mathbf{ h}+\mathbf{v}_f
=\sqrt{P_f}\mathbf{S}_f\bh+\mathbf{A}_f\mathbf{N}_{\widetilde \bh}\dhr+\mathbf{v}_f,\\
\mathbf{z}_f&= \mathbf{X}_f\mathbf{  g}+\mathbf{w}_f
=\sqrt{P_f}\mathbf{S}_f\bg+\mathbf{A}_f\mathbf{N}_{\widetilde \bh}\bg+\mathbf{w}_f, \label{eq forward training_rx_EVE}
\end{align}
where $\mathbf{v}_f$ and $\mathbf{w}_f$ are the AWGN with entries that are
i.i.d. $\mathcal{CN}(0, \sigma^2)$ at the destination and the eavesdropper, respectively.
The destination and the eavesdropper are then able to compute MMSE estimates $\mathbf{\hat h}$ and $\mathbf{\hat g}$ of their respective channels. The channel estimation error vectors are $\mathbf{\Delta h}\triangleq\mathbf{h}-
\mathbf{\hat h}$ and $\mathbf{\Delta g}\triangleq\mathbf{g}-
\mathbf{\hat g}$, whose entries are $0$ mean with
variances
\begin{align}\label{var error h forward training}
\sdh&=\left(\frac{1}{\sh}+\frac{P_fT_f/n_t}{P_{f_a}\sdhr+\sigma^2}\right)^{-1},
\end{align}
and
\begin{align}
\sdg&=\left(\frac{1}{\sg}+\frac{P_fT_f/n_t}{P_{f_a}\sg+\sigma^2}\right)^{-1},\label{var error g forward training}
\end{align}
respectively. The channel estimate $\hat\bh$ is fed back to the source for use in data transmission.

It is interesting to remark that, in the DCE scheme described above,
reverse training is first performed to provide the source with knowledge of the channel between itself and the destination (but does not help the eavesdropper obtain information about its channel from the source). This knowledge is then used by the source to determine the AN placement in the forward training stage so as to minimize its interference at the destination.  In conventional training, only the forward training stage is required since AN is not utilized. In this case, the training length is $T_t=T_f$ (since $T_r=0$)
and the forward training signal can be expressed simply as $\mathbf{X}_f=\sqrt{P_f}\mathbf{S}_f$. Even though the time required for conventional training is less than that of DCE (leaving more channel uses for data transmission in each coherence interval), the achievable secrecy rate may not necessarily be higher due to increased CSI leakage \cite{TYLiu_Mag} to the eavesdropper.


\subsection{Data Transmission Phase - AN-Assisted Secrecy Beamforming}\label{sec.syst.data}

Suppose that the source is able to obtain knowledge of the channel estimate $\hat\bh$ through feedback from the destination but has only statistical knowledge of the eavesdropper's channel $\bg$ (and also $\hat\bg$). Based on this channel knowledge, the source can then utilize in the data transmission phase an AN-assisted secrecy beamforming scheme \cite{SGoel_Negi2008} where the data-bearing signal is directed towards the destination while AN is placed in the null space of $\hat\bh$ to jam the reception at the eavesdropper. The transmit signal is thus given by
\begin{equation}\label{eq_data_signal}
\mathbf{X}_d=\sqrt{P_d}\,\mathbf{s}_d \mathbf{\frac{\hat h^{\dag}}{\|\hat h\|}}+\mathbf{A}_d\mathbf{N}_{ \hat \bh}
\end{equation}
where $\mathbf{s_d} \in \mathcal{C}^{T_d \times 1}$ is
the data-bearing signal vector whose entries are i.i.d. $\mathcal{CN}(0,1)$, $P_d$ is the power of the data signal,
$\mathbf{N}_{\hat \bh} \in \mathcal{C}^{(n_t-1) \times n_t}$ is the
matrix that spans the null space of $\mathbf{\hat h}$ and
satisfies $\mathbf{N}_{\hat \bh}\mathbf{N}^{\dag}_{\hat
\bh}=\mathbf{I}_{n_t-1}$, and $\mathbf{A}_d \in \mathcal{C}^{T_d
\times (n_t-1)}$ is the AN matrix whose entries are
i.i.d. $\mathcal{CN}(0,\frac{P_a}{n_t-1})$. Hence, the total AN power in the data transmission phase is $P_a$.

The signals received at the destination and the eavesdropper are given by
\begin{align} \label{eq data_rx_LR}
\mathbf{y}_d&=\mathbf{X}_d\mathbf{ \hat h}+ \mathbf{X}_d\mathbf{ \Delta h}+\mathbf{v}_d=\sqrt{P_d}\,\mathbf{s}_d \mathbf{\|\hat h\|}+\sqrt{P_d}\,\mathbf{s}_d \mathbf{\frac{\hat h^{\dag}}{\|\hat h\|}}\mathbf{\Delta h}+\mathbf{A}_d\mathbf{N}_{\hat \bh}\mathbf{\Delta h}+\mathbf{v}_d,\\
\mathbf{z}_d&= \mathbf{X}_d\mathbf{ \hat g}+ \mathbf{X}_d \mathbf{\Delta g}+\mathbf{w}_d=\sqrt{P_d}\,\mathbf{s}_d \mathbf{\frac{\hat h^{\dag}}{\|\hat h\|}}\mathbf{\hat g}+\sqrt{P_d}\,\mathbf{s}_d \mathbf{\frac{\hat h^{\dag}}{\|\hat h\|}}\mathbf{\Delta g}+\mathbf{A}_d\mathbf{N}_{\hat \bh}\bg+\mathbf{w}_d,\label{eq_data_rx_EVE}
\end{align}
where
$\mathbf{v}_d\sim\mathcal{CN}({\bf 0}, \sigma^2\bI)$ and $\mathbf{w}_d\sim\mathcal{CN}({\bf 0}, \sigma^2\bI)$ are the AWGN vectors.
The signal and AN powers in both training and data transmission should satisfy the total power constraint
\begin{equation}\label{eq.power_constraint}
(P_rT_r+P_fT_f+P_{f_a}T_f+P_dT_d+P_aT_d)/T\leq P.
\end{equation}

%

\section{Bounds on the Achievable Secrecy Rate with Channel Estimation Error}\label{sec.secrecy rate}

In this work, we are interested in studying the impact of AN in both training and data transmission phases on the achievable secrecy rate of the scheme described in the previous section. In particular, to communicate the confidential message from the source to the destination, we consider a $(2^{nTR},nT)$ wiretap code that spans over the data transmission phases of $n$ coherence intervals. The code consists of an encoder $\phi_n$ that maps the message $W \in \mathcal{W} \triangleq \{1,2,...,2^{nTR}\}$ to a length-$n$ block codeword
$\bs_d^n$ and a decoder
$\psi_n$ that maps the received signal $\by_d^n$ into the message $\hat W\in \mathcal{W}$ at the destination. A secrecy rate $R$ is said to be achievable if there exists a sequence of $(2^{nTR}, nT)$ codes such that
the average error probability at the destination goes to zero, i.e.,
$P_e^{(n)}\triangleq \frac{1}{2^{nTR}} \sum_{w \in \mathcal{W} }\Pr(\hat W \neq w |W=w)\rightarrow 0$,
and the so-called equivocation rate \cite{AKhisti2010_MISOME1,XHe_Yener2014}
converges to the average entropy of $W$, i.e.,
$R_e^{(n)}\triangleq \frac{1}{nT}H(W| \mathbf{z}_d^n,\hat {\mathbf{h}}^n , \hat{\mathbf{g}}^n) \rightarrow \frac{1}{nT}H(W)$,
as the codeword length $n\rightarrow \infty$. Here, $\mathbf{z}_d^n$ is the channel output at
the eavesdropper over $n$ coherence intervals, and $\hat{\mathbf{h}}^n$ and $\hat{\mathbf{g}}^n$ are the
estimated channel vectors
at the destination and eavesdropper, respectively, over the $n$ coherence intervals. The equivocation rate
provides a measure of the information obtained by the eavesdropper and is computed here by conditioning on knowledge of both channel estimates $\hat{\mathbf{h}}^n$ and $\hat{\mathbf{g}}^n$ at the eavesdropper (i.e., a worst case assumption).

Following the results in \cite{ICsiszar1978}, an achievable secrecy rate of the proposed scheme
with imperfect CSI can be written as
\begin{align} \label{eq rate_1}
R =\frac{1}{T}  I(\sd;\yd,\hat \bh)-I(\sd;\zd,\hat \bh,\hat \bg)  =\frac{1}{T}I(\sd;\yd| \hat \bh) -I(\sd;\zd| \hat \bh, \hat \bg),
\end{align}
where the equality follows from the fact that $\mathbf{s}_d$ is independent of $\hat \bh$ and $\hat \bg$. Due to the presence of channel estimation errors, it is difficult to express the achievable secrecy rate in a more explicit form. However, we obtain, in the following theorem, upper and lower bounds that will later be shown to be asymptotically tight at high SNR in the cases under consideration.


\begin{Theo} \label{lem upper lower}
Suppose that channel estimation errors $\Delta\bh$ and $\Delta\bg$ are Gaussian with i.i.d. entries. Then,  for $n_t$ sufficiently large, the achievable secrecy rate $R$ of the AN-assisted secrecy beamforming scheme in Section \ref{sec.syst.data} can be bounded as
\begin{align} \label{eq upper and lower r}
  \tilde{R}- \dRlGen  \leq  R  \leq \tilde{R} + \dRuGen
\end{align}
where 
\begin{align} \label{R lemma 1}
\notag \tilde{R} &\triangleq\frac{T_d}{T}\mathbb{E}\left[ \log \left(1+\frac{\pd(\sh-\sdh)\effh}{\pd \sdh+P_a\sdh+{\sigma^2} }\right)\right]\\
 &~~~~~~~~~~-\frac{T_d}{T}\mathbb{E}\left[ \log \left(1+\frac{\pd(\sg-\sdg)\effg}{\pd \sdg+ P_a(\sg-\sdg)\frac{\ngg}{n_t-1}+P_a\sdg+{\sigma^2} }\right)\right],
\end{align}
\begin{align}
\dRuGen&\triangleq\frac{1}{T}\log \left(\frac{\left(\pd\sdh+P_a\sdh+{\sigma^2}\right)^{T_d}}
 {\left(P_a\sdh+{ \sigma^2}\right)^{T_d-1}}\right)-\frac{1}{T}\mathbb{E}\left[\log\left(\pd \|\sd\|^2\sdh+P_a\sdh+{\sigma^2} \right)\right],\label{Ru lemma 1}
\end{align}
and
\begin{align}
\notag \dRlGen&\triangleq\frac{1}{T}\mathbb{E}\left[\log \left(\frac{\left(\pd\sdg+P_a(\sg-\sdg)\frac{\ngg}{n_t-1}+P_a\sdg+{\sigma^2}\right)^{T_d}}
 {\left(P_a(\sg-\sdg)\frac{\ngg}{n_t-1}+P_a\sdg+{\sigma^2}\right)^{T_d-1}}\right)\right]\\
 &~~~~~~~~~~-\frac{1}{T}\mathbb{E}\left[\log\left(\pd \|\sd\|^2\sdg+P_a(\sg-\sdg)\frac{\ngg}{n_t-1}+P_a\sdg+{\sigma^2}\right)\right].\label{Rl lemma 1}
\end{align}
In the above, $\overline{\mathbf{h}}\triangleq \hat{ \mathbf{h}}/\sqrt{\sh-\sdh}$ and $\overline{\mathbf{g}}\triangleq \hat{ \mathbf{g}}/\sqrt{\sg-\sdg}$ are the normalized channel estimates whose entries are all i.i.d. $\mathcal{CN}(0,1)$. Notice that $\overline{\mathbf{h}}$ and $\overline{\mathbf{g}}$ are normalized so that they are independent of the power allocation, i.e., $P_r, P_f, P_{f_a},P_d$, and $P_a$.
\end{Theo}

Details of the proof can be found in Appendix \ref{sec.proof of bounds}. This theorem shows that the achievable secrecy rate can be bounded around $\tilde R$ given in \eqref{R lemma 1} when the aggregate of the channel estimation error and the AN interference terms are effectively Gaussian. These bounds are analogous to those derived in \cite{BHassibi_Hochwald2003_Training} and \cite{TYoo_Goldsmith} for conventional point-to-point channels.
However, the proof of our theorem requires large $n_t$ analysis to cope with the non-Gaussianity of the additional noise term caused by the combination of estimation error and AN. The bounds in Theorem \ref{lem upper lower} are applicable regardless of the training scheme as long as $\Delta\bh$ and $\Delta\bg$ are Gaussian.
In the following corollary, we show that the bounds are in fact applicable for
both the conventional and the AN-assisted training schemes considered in our work.

\begin{Cor}\label{cor upper and lower bound}
The bounds in Theorem \ref{lem upper lower} hold when either conventional or AN-assisted training (i.e., DCE) schemes with linear MMSE estimation is adopted in the training phase.
\end{Cor}

The corollary can be shown as follows. In the conventional training scheme, no AN interference exists in the received forward training signals in \eqref{eq forward training_rx_LR} and \eqref{eq forward training_rx_EVE} and, thus, the estimation error $\Delta\bh$ (and also $\Delta\bg$) is indeed Gaussian and independent of $\hat\bh$ when employing the linear MMSE estimation (which is also the optimal MMSE estimation in this case) \cite{EstimationTheory}. However, this is not the case in AN-assisted training since the AN interference $\bA_f\bN_{\tilde\bh}\Delta\bh_r$ in \eqref{eq forward training_rx_LR} is non-Gaussian. Yet, by applying Lemma \ref{Lemma_DCE_Gaussian} in Appendix \ref{sec.proof of bounds}, we can also show that $\bA_f\bN_{\tilde\bh}\Delta\bh_r$ is asymptotically Gaussian as $n_t\rightarrow\infty$ since $\Delta\bh_r$ is again Gaussian as a result of the MMSE estimation at the source.
These bounds are utilized in Sections \ref{sec.no AN training} and \ref{Sec.DCE} to derive the optimal power allocation between pilot, data, and AN usage in cases with conventional and AN-assisted training, respectively.

\section{AN-Assisted Secrecy Beamforming with Conventional Training in the High SNR Regime}\label{sec.no AN training}

In this section, we first consider the case where AN is only applied in the data transmission phase, but not in the training phase. We first derive an approximate secrecy rate expression based on the bounds given in the previous section, and use it to determine the optimal power allocation between the  pilot signal in the training phase and the data and AN in the data transmission phase.

\subsection{Asymptotic Approximation of the Achievable Secrecy Rate}

In conventional training (i.e., in the case where AN is not utilized in the training phase), no reverse training is needed and the forward training signal can be written as $\mathbf{X}_f=\sqrt{P_f} \mathbf{S}_f$.
We assume that the training length is equal to the number of transmit antennas, i.e., $T_t=T_f=n_t$, which was shown to be optimal for conventional point-to-point systems without secrecy constraints \cite{BHassibi_Hochwald2003_Training}. Without AN, the signals received at the destination and the eavesdropper in the training phase can be written as
\begin{align} \label{eq forward training_rx_LR no DCE}
\mathbf{y}_f&=\sqrt{P_f}\mathbf{S}_f\mathbf{h}+\mathbf{v}_f,\\
\mathbf{z}_f&=\sqrt{P_f}\mathbf{S}_f\bg+\mathbf{w}_f,\label{eq forward training_rx_EVE no DCE}
\end{align}
By employing MMSE estimation at the destination, the channel estimation error variances in \eqref{var error h forward training} and \eqref{var error g forward training} reduce to
\begin{align}\label{sdh without AN}
\sdh = \frac{\sh{\sigma^2}}{P_f\sh+{\sigma^2}}
\end{align}
and
\begin{align}
\sdg =\frac{\sg{\sigma^2}}{P_f\sg+{\sigma^2}}, \label{sdg without AN}
\end{align}
respectively.
The signal model in the data transmission phase remains the same as in \eqref{eq_data_signal}, \eqref{eq data_rx_LR}, and \eqref{eq_data_rx_EVE}. Let us denote the achievable secrecy rate in this case (i.e., in the case with conventional
training) by $\RANdata$. Then, by
Theorem \ref{lem upper lower}
and Corollary \ref{cor upper and lower bound}, we know that
\begin{equation}
\tRANdata-\dRlANdata \leq \RANdata\leq \tRANdata+\dRuANdata,
\end{equation}
where $\tRANdata$, $\dRlANdata$, and $\dRuANdata$ are given by \eqref{R lemma 1}, \eqref{Ru lemma 1}, and \eqref{Rl lemma 1} with $\sdh$ and $\sdg$ substituted by \eqref{sdh without AN} and \eqref{sdg without AN}.

Let $\mathcal{P}^*(P)\triangleq (P^*_f(P), P^*_d(P), P^*_a(P))$ be the optimal power allocation (i.e., the power allocation that maximizes the achievable secrecy rate $R_{\rm conv}$) under power constraint $P$.
To derive the optimal power allocation, it is often necessary to obtain an explicit expression of the achievable secrecy rate, which is difficult to do in our case as remarked in the previous section. However, we show in the following that the achievable secrecy rate under $\calP^*(P)$, i.e., $\RANdata(\calP^*(P))$, can be closely approximated by $\tRANdata(\calP^*(P))$, for $P$ sufficiently large.
The dependence on $P$ is often neglected in the following for notational simplicity. To express the result, note that two functions $f$ and $g$ are asymptotically equivalent (denoted by $f\doteq g$) if $\lim_{x\rightarrow \infty} f(x)/g(x)=1$.


\begin{Theo}\label{lem small_term} The maximum achievable secrecy rate $\RANdata(\calP^*)$ under conventional training is asymptotically equivalent to $\tRANdata(\calP^*)$ (i.e.,
$\RANdata(\calP^*) \doteq \tRANdata(\calP^*)$)
as $P\rightarrow\infty$.
\end{Theo}

Moreover, we can show that, to achieve the maximum achievable secrecy rate,
the powers assigned to all components, including the pilot in the training phase and the signal and AN in the data transmission phase, should scale at least linearly with $P$ (i.e., should not vanish with respect to $P$ as $P\rightarrow\infty$).
The result can be stated as follows.
\begin{Cor}\label{cor linear ratio CONV}
$P^*_f(P)=\Omega(P)$, $P^*_d(P)=\Omega(P)$, and $P^*_a(P)=\Omega(P)$, where $f(x)=\Omega(g(x))$ denotes the fact that there exists $k_1>0$ such that $k_1 g(x)\leq f(x)$ for all $x$ sufficiently large \cite{IntroToAlgorithms}.
\end{Cor}

The proofs of Theorem \ref{lem small_term} and Corollary \ref{cor linear ratio CONV} can be found in Appendix \ref{app.conv_case_proof}. Notice that, due to the total power constraint in \eqref{eq.power_constraint}, all power components are $O(P)$, where $f(x)=O(g(x))$ indicates that there exists $k_2>0$ such that $f(x)\leq k_2 g(x)$ for all $x$ sufficiently large. That is, all power components increase at most linearly with $P$. Hence, combined with Corollary \ref{cor linear ratio CONV}, it follows that the powers assigned to training, data, and AN should all scale exactly linearly with $P$. In this case,
the channel estimation error variances under $\calP^*$ can be written as $\sdh=\frac{\sh{\sigma^2}}{P_f^*\sh+{\sigma^2} }=  \frac{\sigma^2}{P_f^*}+o\left(\frac{1}{P}\right)$ and $\sdg=\frac{\sg{\sigma^2}}{ P_f^*\sg+{\sigma^2}} = \frac{\sigma^2}{P_f^*}+o\left(\frac{1}{P}\right)$, where $f(x)=o(g(x))$ indicates that $\lim\limits_{x\rightarrow \infty} f(x)/g(x)=0$, and, for $P$ sufficiently large, the achievable secrecy rate can be approximated as
\begin{align}
\notag \tRANdata(\calP^*)\!
=&\frac{T_d}{T}\ee\left[ \log \left(1+\frac{P_d^* (\sh+o(1))\effh}{(P_d^*+P_a^*) \left( \frac{\sigma^2}{P_f^*}\!+\!o\left(\frac{1}{P}\right)\right)\!+\!{\sigma^2} }\right)\right]\\
&~-\frac{T_d}{T}\ee\!\left[ \log\! \left(\!1\!+\!\frac{P_d^*(\sg\!+\!o(1))\effg}{(P_d^*\!+\!P_a^*)\left( \frac{\sigma^2}{P_f^*}\!+\!o\left(\frac{1}{P}\right)\right)\!+\! P_a^* (\sg+o(1))\frac{\|\bN_{\hat h}\mathbf{\overline{g}}\|^2}{n_t-1}\!+\!{\sigma^2} }\!\right)\!\right]\\
=&\frac{T_d}{T}\ee\left[ \log \frac{P_d^*\sh\effh}{\left(\frac{P_d^*+P_a^* }{P_f^* }+1\right){\sigma^2} }\right]
-\frac{T_d}{T}\ee\left[ \log \left(1+\frac{P_d^* \effg}{ P_a^* \frac{\|\bN_{\hat h}\mathbf{\overline{g}}\|^2}{n_t-1} }\right)\right]+o(1).\label{RABF approx used in optimization prob}
\end{align}
This follows from the fact that $(P_d^*(P)+P_a^*(P))/P_f^*(P)=O(1)$ since $P_d^*(P)+P_a^*(P)=O(P)$ by the total power constraint and $P_f^*(P)=\Omega(P)$ by Corollary \ref{cor linear ratio CONV}.

Notice that the approximate secrecy rate given in \eqref{RABF approx used in optimization prob} strictly increases with $P_f^*$, which implies that one can always achieve a higher secrecy rate by increasing the power used for training. This is because the increase of training power benefits the destination by reducing both the effective noise due to channel estimation error and the AN interference; whereas only the channel estimation error is reduced at the eavesdropper. Therefore, the total power constraint should be satisfied with equality at the optimal point, i.e., $P_f^*T_f+P_d^*T_d+P_a^*T_d=PT$.

In fact, for any $\epsilon>0$ and for $n_t$ sufficiently large, it can be further shown that
\begin{align}\label{high_AN_app_line_5}
\tRANdata\geq\frac{T_d}{T}\log \frac{\frac{P_d^* P_a^*}{P_a^*+P_d^*}}{\left(\frac{P_d^*+P_a^* }{P_f^* }+1\right)\sigma^2 }
+\frac{T_d}{T}\mathbb{E}\left[\log \left(\sh\effh(1-\epsilon) \right)\right]
+\frac{T_d}{T}\epsilon_{n_t}+o(1)
\end{align}
The derivations can be found in Appendix \ref{proof of tRconv lower bound}. This lower bound provides an explicit description of the relation between the achievable secrecy rate and the power allocated to each component.

\subsection{Joint Power Allocation between Training and Data Transmission}

In this subsection, we propose a power allocation for the pilot signal, the data signal, and AN with the goal of maximizing the achievable secrecy rate. However, instead of using the achievable secrecy rate $\RANdata$ (whose expression is unknown) as the objective function, we propose a power allocation policy based on the maximization of this lower bound.
More specifically, let us first set $P_a=(PT-P_fT_f-P_dT_d)/T_d$ since the total power constraint must be satisfied. Then, by removing all the terms that are irrelevant to the optimization and by the fact that the logarithm is a monotonically increasing function, we formulate the power allocation problem as follows:
\begin{subequations}\label{opt_problem_AN}
\begin{align}
\max_{P_f, P_d}~~&~~~ \frac{\frac{P_d(PT-P_fT_f-P_dT_d)}
{(PT-P_fT_f)}}{\frac{PT-P_fT_f}{P_fT_d}+1}
\triangleq J_{\rm conv}(P_f, P_d)\\
\mbox{subject to}& ~~~ P_f>0, P_d>0\\
&~~~PT-P_f T_f - P T_d>0.
\end{align}
\end{subequations}
Notice that the powers $P_f$, $P_d$, and  $P_a=(PT-P_fT_f-P_dT_d)/T_d$ are constrained to be greater than zero due to Corollary \ref{cor linear ratio CONV}.

By taking
the first-order derivative of $J_{\rm conv}$ and setting it to zero, we get the solution
\begin{align}\label{eq.opt_conv_approx_sol}
(\hat P_f^*,\hat P_d^*) =\left(\frac{PT\sqrt{T_f}}{T_f\left(\sqrt{T_f}+\sqrt{T_d}\right)},
\frac{PT\sqrt{T_d}}{2T_d\left(\sqrt{T_f}+\sqrt{T_d}\right)}\right).
\end{align}
To verify that $(\hat P_f^*, \hat P_d^*)$ is indeed the optimal solution of \eqref{opt_problem_AN}, it remains to be shown that the Hessian matrix at the
point $(\hat P_f^{*},\hat P_d^{*})$, i.e.,
\begin{align}\label{second order derivate}
\mathbf{H}_{J_{\rm conv}}=\nabla^2 J_{\rm conv} (\hat P_f^{*},\hat P_d^{*}) =
       \left(\begin{array}{ccr}
            -\frac{T_f^2+T_f\sqrt{T_fT_d}}{2PTT_d} &-\frac{T_f}{PT}\\
            -\frac{T_f}{PT} &-\frac{2T_d}{PT}
            \end{array}\right),
\end{align}
is negative semi-definite. Since $\mathbf{H}_{J_{\rm conv}}$ is real and symmetric, this follows from the fact that all principal minors of ${\bf H}_{J_{\rm conv}}$ are positive
i.e.,
\begin{align}
\notag (-1)^1\det \left([\mathbf{H}_{J_{\rm conv}}]_{\{1\},\{1\}}\right)&=-\det \left(-\frac{T_f^2+T_f\sqrt{T_fT_d}}{2PTT_d}\right)>0\\
\notag (-1)^2\det
\left([\mathbf{H}_{J_{\rm conv}}]_{\{1,2\},\{1,2\}}\right)&=\det
\left(\mathbf{H}_{J_{\rm conv}}\right)=\frac{T_f^2+T_f\sqrt{T_fT_d}}{P^2T^2}-\frac{T_f^2}{P^2T^2} >0
\end{align}
due to Sylvester's criterion \cite{MatrixAnalysis}.
Hence,
the proposed power allocation that maximizes the approximate secrecy rate in \eqref{high_AN_app_line_5} is given in the following theorem.

\begin{Prop}
The power allocation that maximizes the approximate secrecy rate in \eqref{high_AN_app_line_5} is
\begin{align}\label{optimal power ratio RConv}
(\hat P_f^*,\hat P_d^*,\hat P_a^*) =\left(\frac{PT\sqrt{T_f}}{T_f\left(\sqrt{T_f}+\sqrt{T_d}\right)},
\frac{PT\sqrt{T_d}}{2T_d\left(\sqrt{T_f}+\sqrt{T_d}\right)},
\frac{PT\sqrt{T_d}}{2T_d\left(\sqrt{T_f}+\sqrt{T_d}\right)}\right).
\end{align}
\end{Prop}

\vspace{.3cm}

The effectiveness of this solution compared to the optimal power allocation $\calP^*$ (i.e., the one that maximizes the achievable secrecy rate $R_{\rm conv}$) will be verified numerically in Section \ref{sec.simulation}. This solution
indicates that, with conventional training, the ratio between the energy used for training and that for data transmission, i.e.,
$\hat P_f^*T_f/(\hat P_d^*T_d+\hat P_a^*T_d)$, should be equal to $\sqrt{T_f/T_d}$. Recall that $T_f$ is equal to $n_t$ whereas $T_d$ increases with the coherence time. Hence, as the coherence time increases, more and more energy should be allocated to the data transmission phase to support the increasing number of channel uses.
Moreover, we can also see from \eqref{optimal power ratio RConv}  that equal power should be allocated to data and AN in the data transmission phase. It is interesting to observe that the solution does not depend on the channel variances $\sigma_h^2$ and $\sigma_g^2$ since, for $P$ sufficiently large, the AWGN terms are negligible and, thus, the SNR at both the destination and the eavesdropper are determined by the ratio between their own received data and AN powers, which experience the same channel gains when arriving at their respective receivers.

Furthermore, by \eqref{RABF approx used in optimization prob}, we can observe that the achievable secrecy rate increases without bound as $P$ increases. However, this is not always the case
when AN is not utilized in either training or data transmission as to be shown in our simulations. This implies that AN is necessary (at least in the data transmission phase) to achieve a secrecy rate that increases without bound with respect to $P$.
However, when the coherence time is large, the energy allocated to training becomes negligible and almost half the total energy is allocated to AN in the data transmission phase (according to \eqref{optimal power ratio RConv}). That is, only half the energy is left to transmit the actual message. However, if AN is further applied in the training phase (as done in
the DCE scheme
\cite{Chang_Chi2010_DCE, CWHuang_Hong2013_TwoWay}),
the difference between the effective channel qualities at the destination and at the eavesdropper can be enhanced, even before the data is actually transmitted. The proportion of AN needed in the data transmission phase can then be reduced.
This is discussed in the following section.



\section{AN-Assisted Secrecy Beamforming with DCE (i.e., AN-Assisted) Training in the High SNR Regime} \label{Sec.DCE}

In this section, we consider the case where AN is used in both the training and the data transmission phases. This refers to the DCE
and the AN-assisted secrecy beamforming schemes described in Sections \ref{sec.syst.training} and \ref{sec.syst.data}, respectively.
Similar to the previous section, we first derive an approximate expression of the achievable secrecy rate and then propose an efficient algorithm for determining the power allocation between pilot, data, and AN in both phases.

\subsection{Asymptotic Approximation of the Achievable Secrecy Rate}


Following Section \ref{sec.syst},
let the length of the reverse and the forward training signals be equal to the number of antennas at the destination and the source, respectively. That is, we set $T_r=1$ and $T_f=n_t$. To distinguish from $\RANdata$ in the previous section, we use $\RANboth$ to denote the achievable secrecy rate of the system considered here.
Similarly by Theorem \ref{lem upper lower}, we can obtain upper and lower bounds of $\RANboth$ as
\begin{equation}
\tRANboth-\dRlANboth\leq \RANboth\leq \tRANboth+\dRuANboth,
\end{equation}
where the terms are given by \eqref{R lemma 1}, \eqref{Ru lemma 1}, and \eqref{Rl lemma 1} with $\sdh$ and $\sdg$ equal to \eqref{var error h forward training} and \eqref{var error g forward training}.

Let $\mathcal{P}^*\triangleq (P^*_r,P^*_f, P^*_{f_a},P^*_d, P^*_a)$ be the optimal power allocation that maximizes the achievable secrecy rate $\RANboth$.
Similar to the case with conventional training, we can also show that $\RANboth(\calP^*)$ can be closely approximated by $\tRANboth(\calP^*)$, for $P$ sufficiently large.

\begin{Theo}\label{asymptotically equivalent DCE}The maximum achievable secrecy rate $\RANboth(\calP^*)$ under DCE training is asymptotically equivalent to $\tRANboth(\calP^*)$ (i.e.,
$\RANboth(\calP^*) \doteq \tRANboth(\calP^*)$)
as $P\rightarrow\infty$.
\end{Theo}

The scaling of the optimal power allocation can also be derived as follows.

\begin{Cor}\label{cor linear ratio DCE}
$P^*_f(P)=\Omega(P)$ and $P^*_d(P)=\Omega(P)$, and that either $P^*_{f_a}(P)=\Omega(P)$ or $P^*_{a}(P)=\Omega(P)$. Moreover, we have $P^*_r(P)=\Omega(P^*_{f_a}(P))$.
\end{Cor}

The proofs of the theorem and the corollary can both be found in Appendix \ref{app.DCE_case_proof}. The corollary shows that, to achieve the maximum acheivable secrecy rate, the power allocated to the forward pilot signal in the training phase and the message-bearing signal in the data transmission phase, i.e., $P^*_f(P)$ and $P^*_d(P)$, should both increase linearly with $P$, and so should the power of at least one of the AN terms (either in the training or data transmission phases, or both). Moreover, the reverse training power $P^*_r(P)$ should scale at least as fast as the AN power $P^*_{fa}(P)$ in the training phase. This is because, with larger AN power $P^*_{fa}(P)$, more power should be invested in reverse training to ensure more accurate placement of AN  in the forward training stage.

By Corollary \ref{cor linear ratio DCE}, the channel estimation error variances in \eqref{var error h forward training} and \eqref{var error g forward training} can be written as
\begin{equation}
\sigma^2_{\Delta h}=\frac{\sh\left(P_{f_a}^*\frac{\sh\sigma^2}{\sigma^2+P_r^*\sh}+\sigma^2\right)}{P_{f_a}^*\frac{\sh\sigma^2}{\sigma^2+P_r^*\sh}+\sigma^2+P_f^*\sh}
= \frac{\sigma^2}{P_f^*}\left(1+\frac{P_{f_a}^*\sh}{\sigma^2+P_r^*\sh}\right)(1+o(1))=o(1),
\end{equation}
since $P_r^*(P)=\Omega(P_{f_a}^*(P))$ and $P_f^*(P)=\Omega(P)$,
and
\begin{equation}\label{eq.g_error_var}
\sigma^2_{\Delta g} =\frac{P_{f_a}^*\sigma^4_{g}+{\sigma^2}\sg}{P_{f_a}^*\sg+ \sigma^2+P_f^*\sg}
=\frac{P_{f_a}^*\sg+{\sigma^2}}{P_{f_a}^*+P_{f}^*}(1+o(1)),
\end{equation}
respectively.
Notice that, in \eqref{eq.g_error_var},
the ratio $\frac{P_{f_a}^*\sg+ \sigma^2}{P_{f_a}^*+P_{f}^*}$ is at least $O(1)$, but may also be $o(1)$ if $P_{f_a}^*$ does not scale as fast as $P_f^*$.
Then, for $P$ sufficiently large, the achievable secrecy rate can be approximated as
\begin{align}
\notag\tRANboth \!=&\frac{T_d}{T}\mathbb{E}\left[ \log \left(1\!+\!\frac{P_d^*(\sh+o(1))\effh}{\left(P_d^*\!+\!P_a^*\right) \frac{\sigma^2}{P_f^*}\left(1\!+\!\frac{P_{f_a}^*\sh}{\sigma^2+P_r^*\sh}\right)(1+o(1))+\sigma^2 }\right)\right]\\
&\!\!-\!\frac{T_d}{T}\mathbb{E}\!\left[ \!\log\! \left(\!1\!+\!\frac{P_d^*\left(\frac{P_f^*\sg}{P_{f_a}^*\!+P_{f}^*}(1+o(1))\right)\effg}{ (P_d^*\!+\!P_a^*)\!\left[\frac{P_{f_a}^*\!\sg\!+\!\sigma^2}{P_{f_a}^*\!+P_{f}^*}(1\!+\!o(1))\!\right]\!+\! P_a^*\!\left[\frac{P_{f}^*\!\sg}{P_{f_a}^*\!+P_f^*}(1\!+\!o(1))\!\right]\!\frac{\ngg}{n_t-1}\!+\!{\sigma^2} }\!\right)\!\right]\!\!\\
 \notag=&\frac{T_d}{T}\mathbb{E}\left[ \log \left(\frac{P_d^*P_f^*(\sigma^2+P_r^*\sh)\sh\effh/\sigma^2}{\left(P_d^*+P_a^*\right){ \left(\sigma^2+P_r^*\sh+P_{f_a}^*\sh\right)}+P_f^*(\sigma^2+P_r^*\sh) }\right)\right]\\
&-\!\!\frac{T_d}{T}\mathbb{E}\!\left[\log \left(\!1\!+\!\frac{P_d^*P_{f}^*\sigma_g^2\effg}{ (P_d^*\!+\!P_a^*)(P_{f_a}^*\!\sigma_g^2\!+\!\sigma^2)\!+\! P_a^*P_f^*\sigma_g^2\frac{\ngg}{n_t-1}\!+\!\sigma^2 (P_{f_a}^*\!+\!P_f^*) }\right)\!\right]\!+\!o(1).\label{high SNR RDCE approximation}
\end{align}

Following the approach used to obtain \eqref{high_AN_app_line_5} (c.f. Appendix \ref{proof of tRconv lower bound}), we can show that, for any $\epsilon'>0$ and for $n_t$ sufficiently large, the second term in \eqref{high SNR RDCE approximation} can be approximated as
\begin{align}
\notag&\frac{T_d}{T}\mathbb{E}\left[\log \left(1+\frac{P_d^*P_f^*\sigma_g^2\effg}{ (P_d^*+P_a^*)(P_{f_a}^*\sigma_g^2+\sigma^2)+ P_a^*P_f^*\sigma_g^2(1-\epsilon')+\sigma^2(P_{f_a}^*+P_f^*)}\right)\right]+\frac{T_d}{T}\epsilon_{n_t}'\\
&=\frac{T_d}{T}\mathbb{E}\left[\log \left(1+\frac{P_d^*P_f^*\sigma_g^2\effg}{ (P_d^*+P_a^*)(P_{f_a}^*\sigma_g^2+\sigma^2)+ P_a^*P_f^*\sigma_g^2(1-\epsilon')}\right)\right]+\frac{T_d}{T}\epsilon_{n_t}'+o(1)\label{DCE second term epsilon form}
\end{align}
where  $\epsilon_{n_t}' \triangleq \mathbb{E}\!\left[\log \left(\!1\!+\!\frac{P_d^*P_{f}^*\sigma_g^2\overline{\bg}^{\dag}\overline{\bh \bh}^{\dag}\overline{\bg}/||\overline{\bh}||^2}{ (P_d^*\!+\!P_a^*)(P_{f_a}^*\!\sigma_g^2\!+\!\sigma^2)\!+\! P_a^*P_f^*\sigma_g^2\ngg/(n_t-1)\!+\!\sigma^2 (P_{f_a}^*\!+\!P_f^*) }\right)\bigg| A^{c}_{\epsilon'}\!\right]\Pr(A^{c}_{\epsilon'}) \rightarrow 0$ as $n_t\rightarrow\infty$ and $A^{c}_{\epsilon'}\triangleq \left\{\left|\ngg/(n_t-1)-1\right|> \epsilon' \right\}$. The equality holds since, by Corollary \ref{cor linear ratio DCE}, either $P^*_{f_a}(P)=\Omega(P)$ or $P^*_{a}(P)=\Omega(P)$. Then, by further applying Jensen's inequality to \eqref{DCE second term epsilon form}, we obtain from \eqref{high SNR RDCE approximation} the following lower bound on $\tRANboth$:
\begin{align}
\notag\tRANboth\geq&  \frac{T_d}{T}\log \left(\frac{P_d^*P_f^*(\sigma^2+P_r^*\sh)}{\left(P_d^*\!+\!P_a^*\right){ \left(\sigma^2\!+\!P_r^*\sh\!+\!P_{f_a}^*\sh\right)}\!+\!P_f^*(\sigma^2\!+\!P_r^*\sh) }\right)+\frac{T_d}{T}\ee\left[\log\frac{\sh\effh}{ \sigma^2}\right]\\
&-\frac{T_d}{T}\log \left(1+\frac{P_d^*P_{f}^*\sigma_g^2}{ (P_d^*\!+\!P_a^*)(P_{f_a}^*\sigma_g^2+\sigma^2)\!+\! P_a^*P_f^*\sigma_g^2(1-\epsilon') }\right)\!+\!\frac{T_d}{T}\epsilon_{n_t}'\!+\!o(1)
\label{RDCE approximation final ver}
\end{align}
It is worthwhile to note that, in this case, the length of the data transmission phase is $T_d=T-T_r-T_f$, which is different from that in the conventional case.

\subsection{Joint Power Allocation between Training and Data Transmission}

Similar to the case with conventional training,  we determine the optimal power allocation by maximizing the lower bound in \eqref{RDCE approximation final ver}. By the fact that the logarithm is a monotonically increasing function and by removing all the terms that are irrelevant to the optimization,
we formulate the power allocation problem as follows:
\begin{subequations}\label{High SNR DCE optimization problem_2}
\begin{align}
\notag \!\!\!\max \limits_{P_r, P_f, P_{f_a},P_d, P_a} &~\frac{P_dP_f(\sigma^2+P_r\sh)}{\left(P_d\!+\!P_a\right)
\left(\sigma^2\!+\!P_r\sh\!+\!P_{f_a}\sh\right)\!+\!P_f(\sigma^2\!+\!P_r\sh) }\\
\notag&~\times\frac{(P_d\!+\!P_a)(P_{f_a}\sigma_g^2\!+\!\sigma^2)\!+\!P_aP_{f}\sigma_g^2(1\!-\!\epsilon')}{(P_d\!+\!P_a)(P_{f_a}\sigma_g^2\!+\!\sigma^2)\!\!+\! P_aP_{f}\sigma_g^2(1\!-\!\epsilon')\!+\!P_dP_{f}\sigma_g^2}\\
&~\triangleq J_{\rm DCE}(P_r,\!P_f,\!P_{f_a},\!P_d,\!P_a)\\
\!\!\!\mbox{subject to}~&P_r>0, P_{f}>0, P_{f_a}>0, P_d>0, P_a>0,\\
& P_r T_r + P_fT_f+P_{f_a} T_f+P_d T_d+P_a T_d=PT.\label{eq.DCEpower_constraint_equalty}
\end{align}
\end{subequations}

Notice that the approximate secrecy rate expression in \eqref{RDCE approximation final ver} follows from
Corollary \ref{cor linear ratio DCE} where it was shown that at least one of the two AN powers (either $P^*_{f_a}(P)$ or $P^*_{a}(P)$, or both) scale linearly with $P$. However, by the proof of Theorem \ref{asymptotically equivalent DCE} in Appendix \ref{app.DCE_case_proof}, we know that the same asymptotic secrecy rate can also be achieved by having all power components $P_r$, $P_f$, $P_{f_a}$, $P_d$, and $P_a$  scale linearly with $P$. In this case, the objective function can be further approximated as
\begin{equation}\label{eq.DCE_approx_obj}
\tilde J_{\rm DCE}(P_r,\!P_f,\!P_{f_a},\!P_d,\!P_a)\!=\!
\frac{P_rP_{f}P_d}
{(P_d\!+\!P_a) (P_{f_a}\!\!+\!P_r)\!+\!P_rP_{f}}\cdot
\!\frac{(P_d\!+\!P_a)P_{f_a}\!\!+\!P_aP_{f}(1\!-\!\epsilon')}{(P_d\!+\!P_a)P_{f_a}\!\!+\! P_aP_{f}(1\!-\!\epsilon')\!+\!P_dP_{f}}.
\end{equation}
Moreover, in \eqref{High SNR DCE optimization problem_2}, the total power constraint is replaced with an equality in \eqref{eq.DCEpower_constraint_equalty} since the objective function increases monotonically with respect to $P_r$ (regardless of whether $J_{\rm DCE}$ or $\tilde J_{\rm DCE}$ is considered). This is because the increase of reverse training power does not benefit the eavesdropper and can be set as large as possible. However, this problem is nonconvex and, thus, is difficult to solve efficiently. To obtain an efficient solution for this problem, we take a successive convex approximation (SCA) approach where we turn the problem into a sequence of geometric programming (GP) problems using the monomial approximation and the condensation method, similar to that done in \cite{MChiang_GeometricProgramming} and \cite{Chang_Chi2010_DCE}. In the following, we describe the procedures of the SCA algorithm briefly using $\tilde J_{\rm DCE}$ as the objective function. The same can be done with $J_{\rm DCE}$ as well. Further details can be found in \cite{MChiang_GeometricProgramming} and \cite{Chang_Chi2010_DCE}.

For convenience, let us consider equivalently the minimization of the inverse of the objective function, i.e.,
\begin{equation}
\tilde J_{\rm DCE}^{-1}(P_r,\!P_f,\!P_{f_a},\!P_d,\!P_a)\!=\!
\frac
{[(P_d\!+\!P_a) (P_{f_a}\!\!+\!P_r)\!+\!P_rP_{f}][(P_d\!+\!P_a)P_{f_a}\!\!+\! P_aP_{f}(1\!-\!\epsilon')\!+\!P_dP_{f}]}{P_rP_{f}P_d[(P_d\!+\!P_a)P_{f_a}\!\!+\!P_aP_{f}(1\!-\!\epsilon')]}.
\end{equation}
Notice that the denominator of $\tilde J_{\rm DCE}^{-1}$ is a posynomial function that can be lower-bounded as
\begin{equation}
P_rP_{f}P_d[(P_d\!+\!P_a)P_{f_a}\!\!+\!P_aP_{f}(1\!-\!\epsilon')]\geq P_rP_{f}P_d\left(\frac{P_dP_{f_a}}{\xi_1}\right)^{\xi_1}
\left(\frac{P_aP_{f_a}}{\xi_2}\right)^{\xi_2}
\left(\frac{(1-\epsilon') P_aP_f}{\xi_3}\right)^{\xi_3}
\end{equation}
for any $\xi_1$, $\xi_2$, $\xi_3\geq 0$, where the right-hand-side is a monomial function. By substituting the term with its monomial lower bound, we obtain a standard GP problem that is solvable in polynomial time. In the SCA algorithm,
this is done iteratively until the solution converges. In particular, suppose that $(P_r^{(i-1)},P_f^{(i-1)},P_{f_a}^{(i-1)},P_d^{(i-1)},P_a^{(i-1)})$ is the solution obtained in the $(i-1)$-th iteration. Then, in the $i$-th iteration, the denominator of $J_{\rm DCE}^{-1}$ is replaced by the monomial function
\begin{equation}
P_rP_{f}P_d\left(\frac{ P_dP_{f_a}}{\xi_1^{(i)}}\right)^{\xi_1^{(i)}}
\left(\frac{ P_aP_{f_a}}{\xi_2^{(i)}}\right)^{\xi_2^{(i)}}
\left(\frac{(1-\epsilon') P_aP_f}{\xi_3^{(i)}}\right)^{\xi_3^{(i)}},
\end{equation}
where $\xi_0^{(i)}=P_d^{(i-1)}P_{f_a}^{(i-1)}+P_a^{(i-1)}P_{f_a}^{(i-1)}+(1-\epsilon') P_a^{(i-1)}P_f^{(i-1)}$, $\xi_1^{(i)}=P_d^{(i-1)}P_{f_a}^{(i-1)}/\xi_0^{(i)}$, $\xi_2^{(i)}=P_a^{(i-1)}P_{f_a}^{(i-1)}/\xi_0^{(i)}$, and $\xi_3^{(i)}=(1-\epsilon') P_a^{(i-1)}P_f^{(i-1)}/\xi_0^{(i)}$. The algorithm is guaranteed to converge to a stationary point of the problem \cite{MChiang_GeometricProgramming}. The procedures are summarized in Algorithm \ref{DCE algorithm} and the resulting solution is denoted by $\hat\calP^*=(\hat P_r^*, \hat P_{f}^*, \hat P_{f_a}^*, \hat P_{d}^*, \hat P_{a}^*)$.


\begin{algorithm}[h]
\caption{Power Allocation for AN-Assisted Training and Data Transmission}
\begin{algorithmic}[1] \label{DCE algorithm}
\STATE {\bf Initialize:} Give an initial set of feasible values $\left(P_r^{(0)},P_f^{(0)},P_{f_a}^{(0)},P_d^{(0)},P_a^{(0)}\right)$ and a convergence threshold $\epsilon_0 > 0$. Set iteration number $i:=0$.

\REPEAT

\STATE $i:=i+1$;

\STATE Set $\xi_0^{(i)}=P_d^{(i-1)}P_{f_a}^{(i-1)}+P_a^{(i-1)}P_{f_a}^{(i-1)}+(1-\epsilon') P_a^{(i-1)}P_f^{(i-1)}$, $\xi_1^{(i)}=P_d^{(i-1)}P_{f_a}^{(i-1)}/\xi_0^{(i)}$, $\xi_2^{(i)}=P_a^{(i-1)}P_{f_a}^{(i-1)}/\xi_0^{(i)}$, and $\xi_3^{(i)}=(1-\epsilon') P_a^{(i-1)}P_f^{(i-1)}/\xi_0^{(i)}$.

\STATE Find $\left(P_r^{(i)},P_f^{(i)},P_{f_a}^{(i)},P_d^{(i)},P_a^{(i)}\right)$ by solving the GP problem
\begin{align*}
\min \limits_{P_r, P_f, P_{f_a}, P_d, P_a}&~~
\frac
{[(P_d\!+\!P_a) (P_{f_a}\!\!+\!P_r)\!+\!P_rP_{f}][(P_d\!+\!P_a)P_{f_a}\!\!+\! P_aP_{f}(1\!-\!\epsilon')\!+\!P_dP_{f}]}{ P_rP_{f}P_d(P_dP_{f_a}/\xi_1^{(i)})^{\xi_1^{(i)}}
({P_aP_{f_a}}/{\xi_2^{(i)}})^{\xi_2^{(i)}}
[{(1-\epsilon') P_aP_f}/{\xi_3^{(i)}}]^{\xi_3^{(i)}}}\\
\mbox{subject to}~&~~~P_r>0, P_{f}>0, P_{f_a}>0, P_d>0, P_a>0,\\
& ~~~P_r T_r + P_fT_f+P_{f_a} T_f+P_d T_d+P_a T_d= PT.
\end{align*}

\UNTIL $
   \frac{J_{\rm DCE}^{-1}\left(P_r^{(i)},P_f^{(i)},P_{f_a}^{(i)},P_d^{(i)},P_a^{(i)}\right)
   -J_{\rm DCE}^{-1}\left(P_r^{(i-1)},P_f^{(i-1)},P_{f_a}^{(i-1)},P_d^{(i-1)},P_a^{(i-1)}\right) }
   {J_{\rm DCE}^{-1}\left(P_r^{(i-1)},P_f^{(i-1)},P_{f_a}^{(i-1)},P_d^{(i-1)},P_a^{(i-1)}\right) }  < \epsilon_0
$.

\STATE {\bf Output} $(\hat P_r^*, \hat P_{f}^*, \hat P_{f_a}^*, \hat P_{d}^*, \hat P_{a}^*)=\left(P_r^{(i)},P_f^{(i)},P_{f_a}^{(i)},P_d^{(i)},P_a^{(i)}\right)$ .
\end{algorithmic}
\end{algorithm}

\subsection{Comparison with the Conventional Training Case}

It is worthwhile to remark that, compared to the conventional training scheme in the previous section, the DCE scheme requires an additional symbol period in the training phase for reverse training. This results in a smaller pre-log factor and, thus, a significant loss in secrecy rate at high SNR. However, in the following corollary, we show that the DCE scheme can always achieve a higher secrecy rate as long as the coherence time is sufficiently long.



\begin{Cor}\label{lem sufficient condition}
Let $\hat\calP^*_{\rm conv}$ be the solution given in \eqref{optimal power ratio RConv}. Then, for $P$ and $n_t$ sufficiently large, there exists $\calP=(P_r, P_{f_a}, P_f, P_d, P_a)$ such that $\tRANboth(\calP) > \tRANdata(\hat\calP^*_{\rm conv})$ if
\begin{align}\label{sufficient condition of T}
T\geq \max\left\{\frac{(4n_t+10)^2}{n_t}, 22\log_{10}\left(\frac{P\sh n_t}{ 4\sigma^2}\right)+1\right\}+n_t.
\end{align}
\end{Cor}

\vspace{.3cm}

The proof can be found in Appendix \ref{sec. proof sufficient condition DCE is better}. Corollary \ref{lem sufficient condition} implies that the DCE scheme can outperform conventional training when $T$ is sufficiently large, even though an additional channel use is occupied by reverse training. This is because, with DCE training, the effective channel qualities at the destination and the eavesdropper are already well-discriminated in the training phase and thus a larger portion of energy can be allocated to data rather than AN in the data transmission. Therefore, the achievable secrecy rate of the DCE scheme increases faster than that of conventional training as the coherence time increases.
Note that Corollary \ref{lem sufficient condition} provides only a sufficient condition on the coherence time $T$. The advantage of DCE can actually be observed for much smaller values of $T$ as shown in our simulations.


\section{Secrecy Rate in the Low SNR Regime}\label{Sec.low power regime}

In this section, we examine the achievable secrecy rate and the corresponding optimal power allocation in the low SNR regime, i.e., in the case where $\sigma^2\rightarrow \infty$.

Let ${\bf u}_r(\hat\bh)\triangleq \sqrt{P_d}\bs_d \|\hat\bh\|+\sqrt{P_d}\bs_d\frac{\hat\bh^\dagger}{\|\hat\bh\|}\Delta\bh+\bA_d\bN_{\hat \bh}\Delta\bh$ be the summation of all terms other than the AWGN in $\by_d$ of \eqref{eq data_rx_LR}. Then,
we have
\begin{align}
I(\sd;\by_d|\hat \bh)&=\int_{\underline{\hat h}}f(\underline{\hat h}) I(\sd;\by_d|\hat \bh=\underline{\hat h})d\underline{\hat h}\\
&=\int_{\underline{\hat h}} f(\underline{\hat h}) I(\sd;\mathbf{u}_r(\underline{\hat h})+\bv_d)d\underline{\hat h}\\
&=\int_{\underline{\hat h}} f(\underline{\hat h}) \left[\frac{\log e}{\sigma^2}G(\sd,\mathbf{u}_r(\underline{\hat h}))+\frac{\log e}{2\sigma^{4}}\Delta(\sd,\mathbf{u}_r(\underline{\hat h}))+o(\sigma^{-4})\right]d\underline{\hat h},\label{mutual info low SNR approx 1}
\end{align}
where $G(\mathbf{x},\mathbf{y})\triangleq\ee\left[\|\ee[\mathbf{y}|\mathbf{x}]-\ee[\mathbf{y}]\|^2\right]$ and $\Delta(\mathbf{x};\mathbf{y}) \triangleq \tr \{\ee\left[{\rm cov}^2(\mathbf{y}|\mathbf{x})\right]-{\rm cov}^2(\mathbf{y})\}$. The equality in \eqref{mutual info low SNR approx 1} follows from the asymptotic expression of the mutual information given in \cite[Theorem $1$]{VVPrelov_Verdu2004_LowSNR}.
By direct calculation of $G(\sd,\mathbf{u}_r(\underline{\hat h}))$ and $\Delta(\sd,\mathbf{u}_r(\underline{\hat h}))$, and by taking the expectation over $\underline{\hat h}$, it can be shown that
\begin{align}\label{low power mutual LR}
I(\sd;\by_d|\hat \bh)
&=\frac{\log e}{\sigma^4}\left(P_dT_dP_fT_f\shh+P_d^2\sdhh T_d^2/2\right)+o(\sigma^{-4}).
\end{align}
Similarly, we have
\begin{align}\label{low power mutual UR}
I(\sd;\bz_d|\hat \bh, \bg)
&=\frac{\log e}{\sigma^4}\left(P_dT_dP_fT_f\sgg/n_t+P_d^2\sdgg T_d^2/2\right)+o(\sigma^{-4}).
\end{align}
Notice that the first term in \eqref{low power mutual LR} is larger than that in \eqref{low power mutual UR} by a factor of $n_t$ due to the processing gain provided by transmit beamforming. By combining the above,
we obtain the following result.
\begin{Theo}
In the low SNR regime, the achievable secrecy rate of the training-based transmission schemes is
\begin{align}\label{low power secrecy rate}
R_s=\frac{\log e}{T\sigma^4}\left(P_dT_dP_fT_f(\shh-\sgg/n_t)+P_d^2 T_d^2(\sdhh-\sdgg)/2\right)+o(\sigma^{-4}).
\end{align}
\end{Theo}

\vspace{.2cm}

Notice that the above secrecy rate does not depend on the AN power $P_a$ in the data transmission phase, and that $\sdh\rightarrow \sh$ and $\sdg\rightarrow \sg$ as $\sigma^2\rightarrow\infty$ regardless of the AN power $P_{fa}$ in the training phase. Hence, the same asymptotic secrecy rate can be achieved even without the use of AN
and, thus, all power can be allocated to the transmission of either the pilot or the data signals.
However, it should be noted that the secrecy rate in \eqref{low power secrecy rate} decays as $1/\sigma^4$ which is much worse than that achievable when the noncoherent transmission scheme, previously proposed in \cite{CRao_Hassibi2004_LowSNR} for conventional point-to-point channels (without secrecy constraints), is employed. In fact, by directly applying the transmission scheme in \cite{CRao_Hassibi2004_LowSNR} to the wiretap channel model under consideration, we can achieve a secrecy rate that  decays only as $1/\sigma^2$. This is because the secrecy beamforming and AN-assisted training and data transmission schemes considered in this paper all rely on accurate channel knowledge, which is difficult to obtain at low SNR, whereas the noncoherent transmission scheme in \cite{CRao_Hassibi2004_LowSNR} does not. This shows that one can actually do better without training in the low SNR regime.


\section{Numerical Results}\label{sec.simulation}

\begin{figure}[t]
\begin{center}
  \includegraphics[width=12cm]{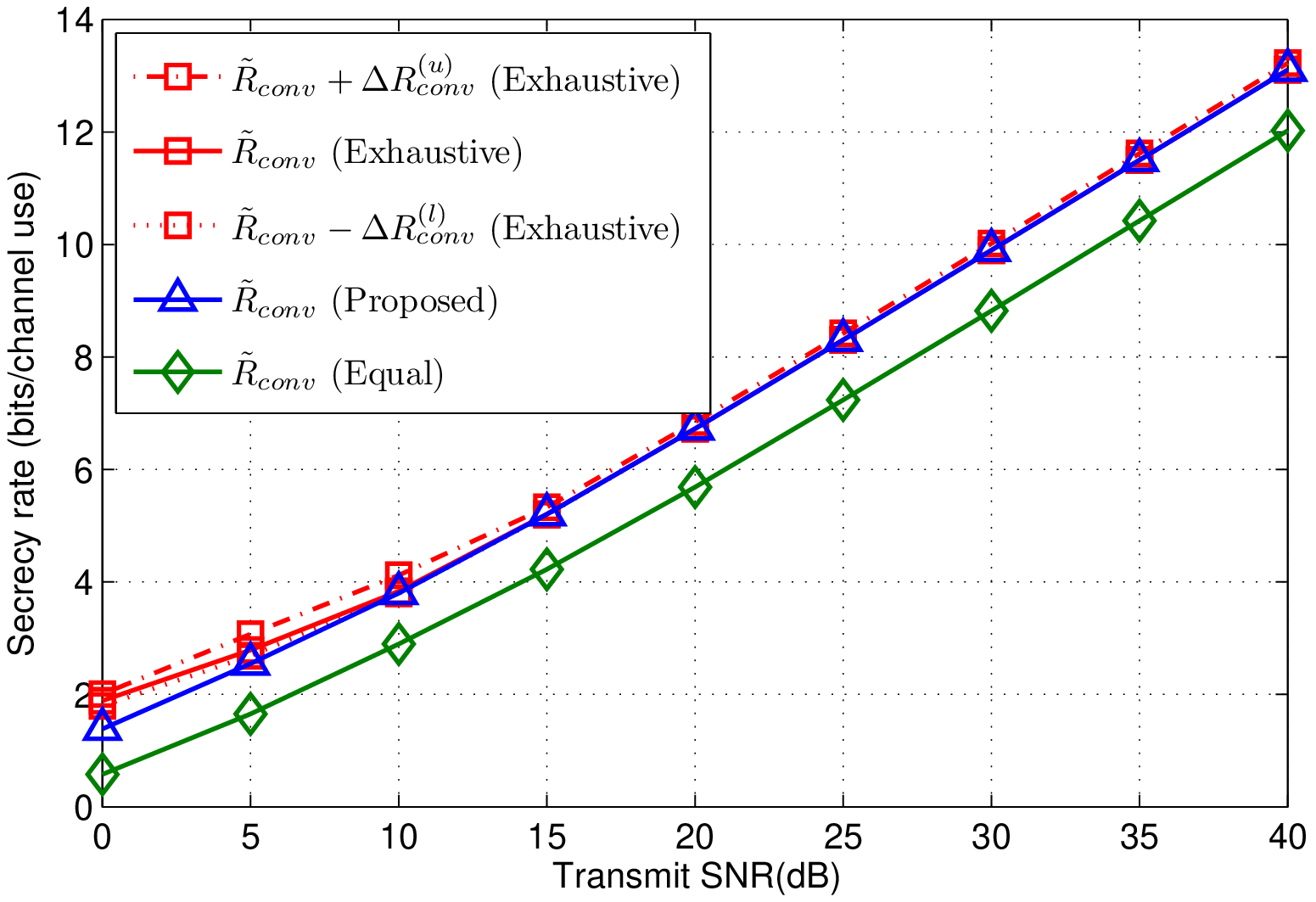}\\
    \vspace{-0.1in}
  \caption{The achievable secrecy rate $\tRANdata$ with different power allocations versus SNR.
   }\label{fig:highAN}
\end{center}
\vspace{-0.2in}
\end{figure}

In this section, we verify numerically our theoretical claims and compare the achievable secrecy rates of different training and power allocation schemes. Unless mentioned otherwise, the number of antennas at the source is $n_t=16$, the coherence interval is $T=480$, the forward training length is $T_f=16$, and the reverse training length is $T_r=1$ (when considering the DCE scheme). The transmit SNR is defined as $P/\sigma^2$ and the channel variances are
$\sh=\sg=0.5$.

In Fig. \ref{fig:highAN}, we show the approximate achievable secrecy rate $\tilde R_{\rm conv}(\hat\calP^*_{\rm conv})$ of the conventional training case with $\hat\calP^*_{\rm conv}$ being the proposed power allocation given in \eqref{optimal power ratio RConv} (labeled as ``$\tilde R_{\rm conv} \text{ (Proposed)}$'') and compare it with the maximum value $\max\limits_{\calP}\tilde R_{\rm conv}(\calP)$ obtained via exhaustive search (i.e., ``$\tilde R_{\rm conv} \text{ (Exhaustive)}$''). We can see that the approximate solution given in \eqref{optimal power ratio RConv} is indeed near optimal at high SNR and yields about $4$ dB improvement over the case with equal power allocation among all components (i.e., ``$\tilde R_{\rm conv} \text{ (Equal)}$''). Moreover, by comparing $\tilde R_{\rm conv}(\hat\calP^*_{\rm conv})$ with the optimized upper and lower bounds $\max\limits_{\calP} \{\tilde R_{\rm conv}(\calP)+\Delta R_{\rm conv}^{(u)}(\calP)\}$ and $\max\limits_{\calP} \{\tilde R_{\rm conv}(\calP)-\Delta R_{\rm conv}^{(l)}(\calP)\}$, respectively, (i.e., ``$\tilde R_{\rm conv}+\Delta R_{\rm conv}^{(u)} \text{ (Exhaustive)}$'' and ``$\tilde R_{\rm conv}-\Delta R_{\rm conv}^{(l)} \text{ (Exhaustive)}$''), we can also see that the approximate secrecy rate expression $\tilde R_{\rm conv}(\hat\calP^*_{\rm conv})$ indeed closely approximates the maximum achievable secrecy rate $R_{\rm conv}(\calP^*_{\rm conv})$ (i.e., Theorem \ref{lem small_term}), where $\calP^*_{\rm conv}$ is the power allocation that maximizes $R_{\rm conv}$,
since $\max\limits_{\calP} \{\tilde R_{\rm conv}(\calP)-\Delta R_{\rm conv}^{(l)}(\calP)\}\leq
R_{\rm conv}(\calP^*_{\rm conv})
\leq \max\limits_{\calP} \{\tilde R_{\rm conv}(\calP)+\Delta R_{\rm conv}^{(u)}(\calP)\}$ and $\tilde R_{\rm conv}(\hat\calP^*_{\rm conv})\approx\max\limits_{\calP} \{\tilde R_{\rm conv}(\calP)-\Delta R_{\rm conv}^{(l)}(\calP)\}\approx \max\limits_{\calP} \{\tilde R_{\rm conv}(\calP)+\Delta R_{\rm conv}^{(u)}(\calP)\}$ at high SNR, as shown in  Fig. \ref{fig:highAN}.

\begin{figure}[t]
\begin{center}
  \includegraphics[width=12cm]{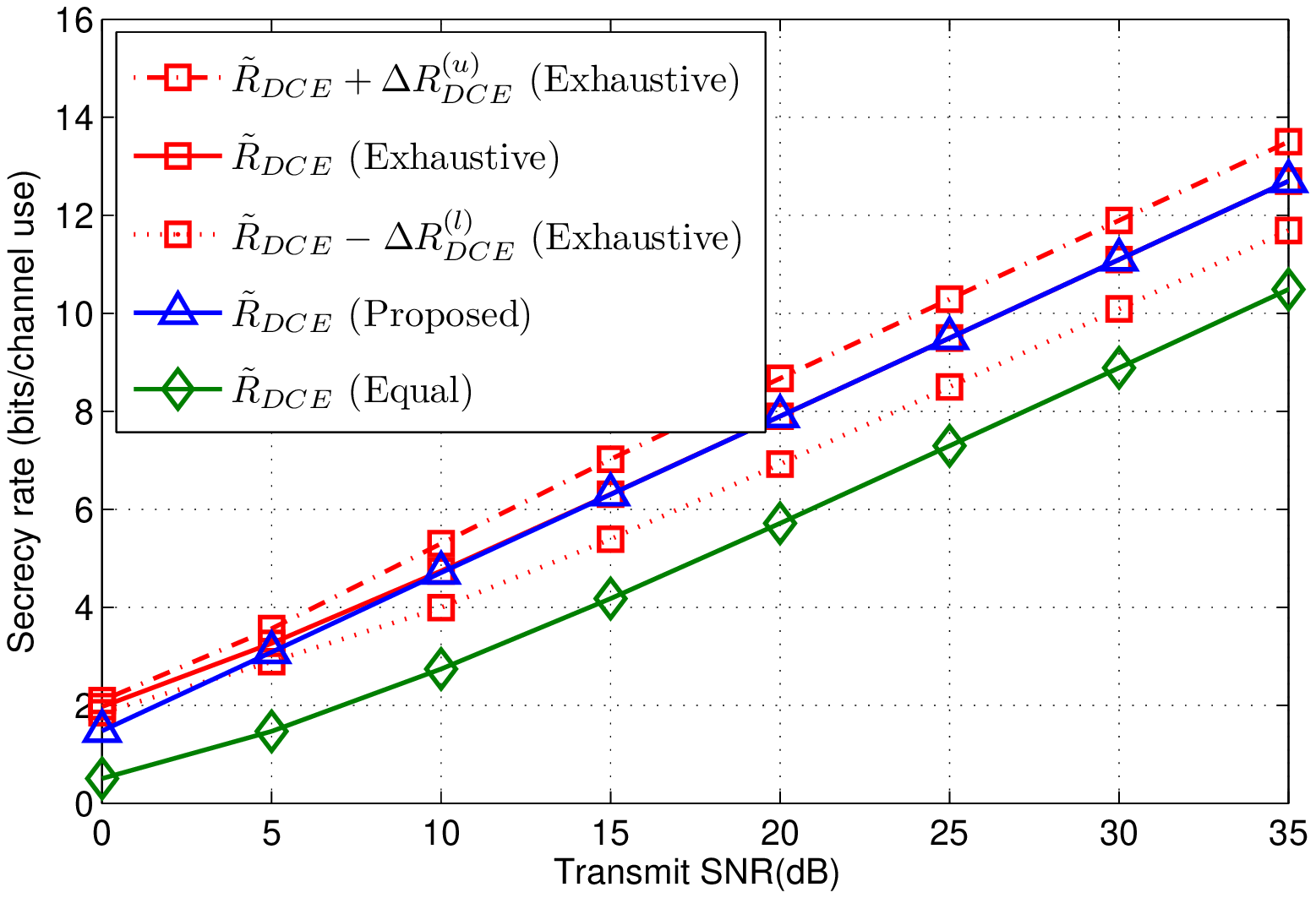}\\
    \vspace{-0.1in}
  \caption{The achievable secrecy rate $\tRANboth$ with different power allocations versus SNR.
   }\label{fig:R_DCE}
\end{center}
\vspace{-0.2in}
\end{figure}

In Fig. \ref{fig:R_DCE}, we show the approximate achievable secrecy rate $\tilde R_{\rm DCE}(\hat\calP^*_{\rm DCE})$ of the DCE training case with $\hat\calP^*_{\rm DCE}$ being the proposed solution obtained by Algorithm \ref{DCE algorithm}  (i.e. ``$\tilde R_{\rm DCE} \text{ (Proposed)}$'') and compare it with the maximum value $\max\limits_{\calP}\tilde R_{\rm DCE}(\calP)$ obtained via exhaustive search  (i.e., ``$\tilde R_{\rm DCE} \text{ (Exhaustive)}$''). Again, the secrecy rate obtained with the proposed solution rapidly converges towards the maximum value obtained via exhaustive search as the transmit SNR increases.  A $7$ dB improvement is also observed when compared to the case with equal power allocation. Moreover, since the optimized upper and lower bounds $\max\limits_{\calP} \{\tilde R_{\rm DCE}(\calP)+\Delta R_{\rm DCE}^{(u)}(\calP)\}$ and $\max\limits_{\calP} \{\tilde R_{\rm DCE}(\calP)-\Delta R_{\rm DCE}^{(l)}(\calP)\}$ maintains a constant difference as the transmit SNR increases and, by Fig. \ref{fig:R_DCE}, $\tilde R_{\rm DCE}(\hat\calP^*_{\rm DCE})$ maintains between the two bounds,  it follows that the difference between the approximate and the actual rates, i.e., $\tilde R_{\rm DCE}(\hat\calP^*_{\rm DCE})-R_{\rm DCE}(\calP^*_{\rm DCE})$, where $\calP^*_{\rm DCE}$ is the power allocation that maximizes $R_{\rm DCE}$, becomes negligible compared to $R_{\rm DCE}(\calP_{\rm DCE}^*)$.


\begin{figure}[t]
\begin{center}
  \includegraphics[width=12cm]{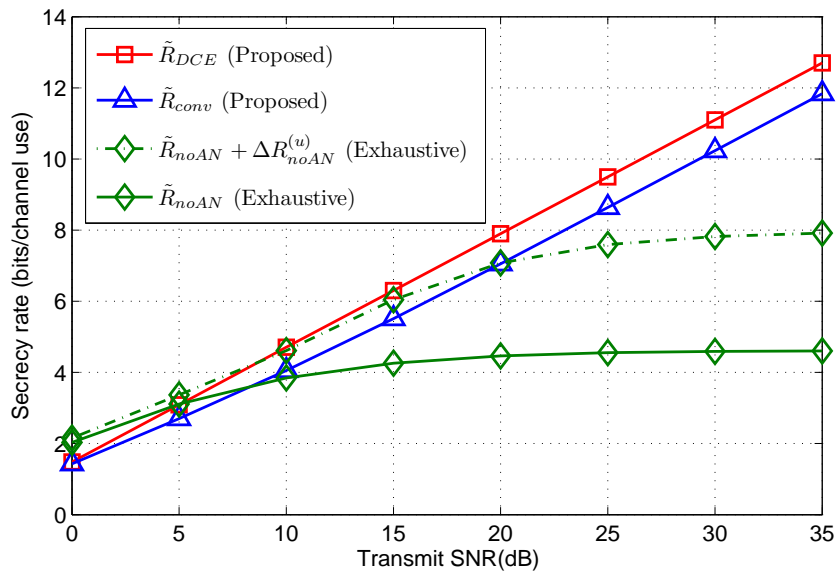}\\
    \vspace{-0.1in}
  \caption{The achievable secrecy rate with different schemes versus SNR.
   }\label{fig:DCE}
\end{center}
\vspace{-0.2in}
\end{figure}

In Fig. \ref{fig:DCE}, we compare the (approximate) achievable secrecy rate of different transmission schemes, namely, the case with conventional training (i.e., the case where AN is utilized only in the data transmission phase), the case with DCE training, and the case where no AN is used in either training or data transmission. Recall that $T_d=T-T_f-T_r$ where $T_r=1$ in the case with DCE training and is $0$ otherwise. We can observe that DCE training yields the best performance even though an additional channel use is required for reverse training. Moreover, we can see that, when AN is not used in either training and data transmission, the achievable secrecy rate becomes bounded as the transmit SNR increases, regardless of whether we are looking at $\tilde R_{\rm noAN}$ or the upper bound $\tilde R_{\rm noAN}+\Delta R_{\rm noAN}^{(u)}$. This indicates that the use of AN is critical to achieve good secrecy rate performance in the high SNR regime.

%

\begin{figure}[t]
\begin{center}
  \includegraphics[width=12cm]{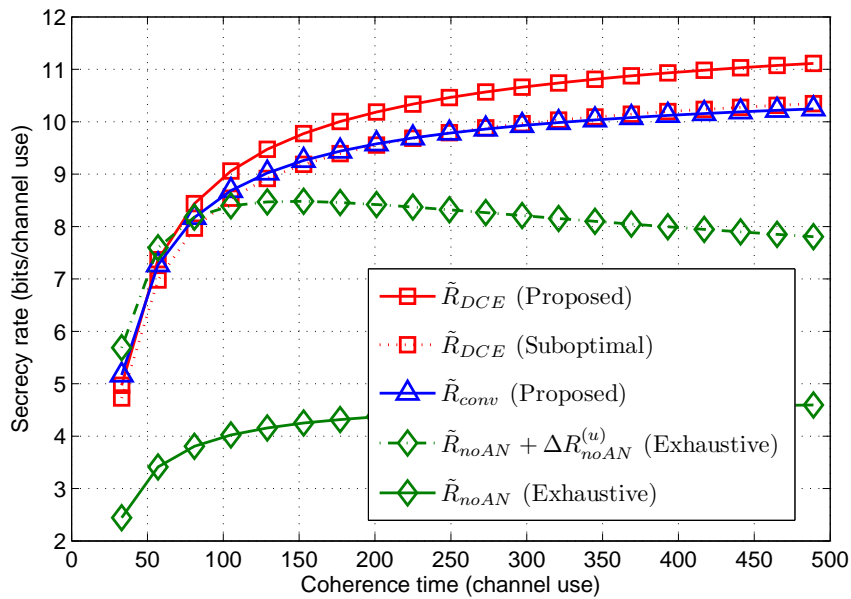}\\
  \vspace{-0.1in}
  \caption{The achievable secrecy rate with different schemes versus coherence time.
   }\label{fig:Rate with different schemes}
\end{center}
\vspace{-0.2in}
\end{figure}

In Fig. \ref{fig:Rate with different schemes}, we verify the effect of coherence time on the achievable secrecy rate of the different schemes. Here, the transmit SNR is fixed at $30$ dB. The DCE scheme with suboptimal power allocation refers to the power allocation used to prove the sufficient condition in Corollary \ref{lem sufficient condition}. The suboptimal solution performs significantly worse than the proposed solution, but was sufficient to yield the condition in Corollary \ref{lem sufficient condition}.
In fact, with the proposed power allocation, DCE is able to outperform conventional training with a coherence time of only $70$, which is
considerably smaller than the value $250$ required by the suboptimal power allocation. Yet, the latter is still smaller than the value $358.25$ predicted by Corollary \ref{lem sufficient condition}, where the result is more conservative.
Moreover, by comparing between ``$\tilde R_{\rm DCE}$ (Proposed)'' and ``$\tilde R_{\rm conv}$ (Proposed)'', we can also see that the advantage of utilizing AN in the training phase increases as the coherence time increases. This is because, by applying AN in the training phase, we can allocate less energy to AN in the data transmission phase and, thus, more energy to the actual message-bearing signal.



\section{Conclusion}\label{sec.conclusion}

In this paper, we examined the impact of both conventional and AN-assisted training on the achievable secrecy rate of the AN-assisted secrecy beamforming scheme. Bounds on the achievable secrecy rate were first derived and then utilized to obtain a closed-form approximation that is shown to be asymptotically tight at high SNR. The approximate expression was then adopted as the objective function to determine the power allocation between pilot signals, data signals, and AN in both training and data transmission phases. An asymptotically optimal closed-form solution was obtained for the case with conventional training whereas a successive convex approximation approach was proposed for the case with DCE training. Furthermore, in the low SNR regime, we showed that AN  provides no gains in secrecy rate  and, thus, is not needed   in either training or data transmission. Numerical simulations were provided to verify the tightness of the bounds and the advantages of DCE over conventional training.



\begin{appendices}
\section{Proof of Theorem \ref{lem upper lower}}\label{sec.proof of bounds}

Here, we first derive upper and lower bounds of $I(\sd;\yd| \hat \bh)$ and $I(\sd;\zd| \hat \bh, \hat \bg)$, and apply them directly to obtain the desired bounds for $R$, which is the difference of the two quantities.
The derivations of the upper and lower bounds are shown only for
$I(\sd;\yd| \hat \bh)$ whereas that of
$I(\sd;\zd| \hat \bh, \hat \bg)$ can be obtained similarly.

\subsection{Lower Bound of $I(\sd;\yd| \hat \bh)$}

To derive the lower bound of $I(\sd;\yd| \hat \bh)$, let us write
\begin{equation}\label{eq.Isd_yd}
I(\sd;\yd| \hat \bh)=h(\sd|\hat \bh)-h(\sd|\yd,\hat \bh),
\end{equation}
where $h(\sd|\hat\bh)=h(\sd)=T_d\log(\pi e)$ and $h(\sd|\yd,\hat\bh)\leq \ee_{\hat \bh,\yd}\left[\log\left((\pi e)^{\td}\left| \mathbf{C}_{\sd|\yd,\hat \bh}\right|\right)\right]$ since Gaussian maximizes entropy. Here, $\mathbf{C}_{\mathbf{a}|\mathbf{b}}$ represents the covariance matrix of $\mathbf{a}$ given $\mathbf{b}$, and $|\bA|$ represents the determinant of $\bA$. Moreover, for any estimate $\hat{\mathbf{s}}_d$ of $\sd$ given  $\yd$ and $\hat \bh$,  we have
$\mathbf{C}_{\sd|\yd,\hat \bh} \preccurlyeq \ee_{\sd|\yd,\hat \bh} \left[ (\sd -\hat{\mathbf{s}}_d)(\sd -\hat{\mathbf{s}}_d)^{\dag}\right]$,
where $A\preccurlyeq B$ denotes that $A-B$ is semi-negative definite, and thus
$\left|\mathbf{C}_{\sd|\yd,\hat \bh}\right| \leq \left|\ee_{\sd|\yd,\hat \bh} \left[ (\sd -\hat{\mathbf{s}}_d)(\sd -\hat{\mathbf{s}}_d)^{\dag}\right]\right|$.
Therefore, for $\hat{\mathbf{s}}_{d}^{L}=\mathbf{C}_{\sd\yd|\hat \bh}\mathbf{C}_{\yd|\hat \bh}^{-1}\yd$ (i.e., the LMMSE of $\sd$ given $\by_d$ while assuming that $\hat\bh$ is known), we have
\begin{align}
&\ee_{\hat \bh,\yd}\!\left[\log\left((\pi e)^{\td}\left| \mathbf{C}_{\sd|\yd,\hat \bh}\right|\right)\right]\notag\\
&\leq \ee_{\hat \bh}\left[\ee_{\yd|\hat \bh}\left[\log\left((\pi e)^{\td}\left| \ee_{\sd|\yd,\hat \bh}\left[(\sd-\hat{\mathbf{s}}_{d}^{L})(\sd-\hat{\mathbf{s}}_{d}^{L})^{\dag}\right]\right|\right)\right]\right]\\
&\leq \ee_{\hat \bh}\left[\log\left((\pi e)^{\td}\left| \mathbf{C}_{\sd|\hat \bh}-\mathbf{C}_{\sd\yd|\hat \bh}\mathbf{C}_{\yd|\hat \bh}^{-1}\mathbf{C}_{\yd\sd|\hat\bh} \right|\right)\right]\\
&= \ee_{\hat \bh}\!\left[\log\left((\pi e)^{\td}\left|\bI_{T_d}-\frac{P_d\|\hat \bh\|^2}{\pd\|\hat\bh\|^2+\pd\sdh+P_a\sdh+{\sigma^2}}\bI_{T_d} \right|\right)\right],\label{lower bound proof without AN 2}
\end{align}
where the last inequality follows from Jensen's inequality. Hence, by combining \eqref{eq.Isd_yd} and \eqref{lower bound proof without AN 2}, we have
\begin{equation}
I(\sd;\yd| \hat \bh)\geq
\td \ee_{\hat \bh}\left[ \log \left(1+\frac{\pd\hh}{\pd \sdh+P_a\sdh+{\sigma^2} }\right)\right].
\label{lower bound of mutual information of LR without AN}
\end{equation}

\subsection{Upper Bound of $I(\sd;\yd| \hat \bh)$}

To obtain the upper bound, we instead write
\begin{align} \label{upper bound proof without AN}
I(\sd;\yd| \hat \bh)=h(\yd|\hat \bh)-h(\yd|\sd,\hat \bh)
\end{align}
where $h(\yd|\hat \bh)\leq \ee_{\hat\bh}\left[\log\left((\pi e)^{\td}\left| \mathbf{C}_{\yd|\hat \bh}\right|\right)\right]$ since Gaussian maximizes entropy and
$h(\yd|\sd,\hat \bh)=h(\sqrt{P_d}\,\mathbf{s}_d \mathbf{\frac{\hat h^{\dag}}{\|\hat h\|}}\mathbf{\Delta h}+\mathbf{A}_d\mathbf{N}_{\hat \bh}\mathbf{\Delta h}+\mathbf{v}_d|\sd,\hat\bh)$ by \eqref{eq data_rx_LR}. Notice that $h(\yd|\sd,\hat \bh)$ is difficult to evaluate since $\mathbf{A}_d\mathbf{N}_{\hat \bh}\mathbf{\Delta h}$ is non-Gaussian. Hence, we resort to the following large $n_t$ analysis.


\begin{Lem} \label{Lemma_DCE_Gaussian}
Let $\bA$ be a $t\times (n-1)$ matrix with entries being i.i.d. ${\cal CN}\left(0,\frac{P}{n-1}\right)$, $\bN$ be a $(n-1) \times n$ semi-unitary matrix such that $\bN\bN^\dagger=\bI$, and $\Delta\bh$ be an $n\times 1$ vector with entries being i.i.d. ${\cal CN}(0,\sigma_{\Delta h}^2)$.
Then, $\bA\bN\Delta\bh$ converges in distribution to a Gaussian vector with entries being i.i.d. ${\cal CN}(0,P\sdh)$
as $n\rightarrow\infty$.
\end{Lem}

\begin{proof}
Let $\{\bN\}_{i,j}$ denote the $(i,j)$-th entry of matrix $\bN$ and let $\Delta\bh_j$ denote the $j$-entry of vector $\Delta\bh$. Then, we can define the vector ${\bf b}\triangleq \bN\Delta\bh$ whose $i$-th entry can be written as ${\bf b}_i=\sum_j\{\bN\}_{i,j}\Delta\bh_j$. Note that $\bf b$ is a Gaussian vector with entries that are i.i.d. with mean $0$ and variance $\sigma_{\Delta h}^2$ since, regardless of the value of $\bN$,  $\ee[{\bf b}_i{\bf b}_k^*|\bN]=\sum_j\sum_\ell \{\bN\}_{i,j}\{\bN\}_{k,\ell}^* \ee\left[\Delta\bh_j\Delta\bh_\ell^*\right]=\sigma_{\Delta h}^2$, for $i=k$, and $0$, otherwise.
Then, it follows by central limit theorem that the $i$-th entry of vector  $\bA\bN\Delta\bh=\bA{\bf b}$, i.e., $\sum_j\{\bA\}_{i,j}{\bf b}_j=\frac{1}{\sqrt{(n-1)}}\sum_j\{\overline{\bA}\}_{i,j}{\bf b}_j$ where $\{\overline{\bA}\}_{i,j}\triangleq\sqrt{n-1}\{\bA\}_{i,j}\sim {\cal CN}\left(0,P\right)$, converges in distribution to a Gaussian random variable with mean $0$ and variance $P\sdh$, as
$n \rightarrow\infty$.
Moreover, since $\sum_j\sum_\ell\ee[\{\bA\}_{i,j}{\bf b}_j\{\bA\}_{k,\ell}{\bf b}_\ell]=0$
for $i\neq k$ (i.e., the entries of $\bA{\bf b}$ are uncorrelated), it follows that $\bA{\bf b}$ converges in distribution to a Gaussian vector with entries that are i.i.d. ${\cal CN}(0,P\sdh)$, as $n\rightarrow\infty$.
\end{proof}

By Lemma \ref{Lemma_DCE_Gaussian} (with $n=n_t$, $t=T_d$, and $P=P_a$), we know that $\mathbf{A}_d\mathbf{N}_{\hat \bh}\mathbf{\Delta h}$ is asymptotically Gaussian as $n_t\rightarrow\infty$ if $\Delta\bh$ is Gaussian as well. Hence, for $n_t$ sufficiently large, we have
\begin{align}
I(\sd;\yd| \hat \bh)&\leq \ee_{\hat\bh}\left[\log\left((\pi e)^{\td}\left| \mathbf{C}_{\yd|\hat \bh}\right|\right)\right]-\ee_{\hat\bh}\left[\log\left((\pi e)^{\td}\left|\mathbf{C}_{\yd|\sd,\hat \bh}\right|\right)\right]\\
 &
=\ee_{\hat\bh}\left[\log\frac{\left|\left(\!P_d\hh\!+\!P_d\sdh\!+\!P_a\sdh\!+\!{ \sigma^2}\!\right)\mathbf{I}_{T_d}\right|}
{\left|P_d\sdh\sd\sd^{\dag}\!+\!(P_a\sdh\!+\!{ \sigma^2})\mathbf{I}_{T_d}\right|}\!\right]\\
&=\ee_{\hat\bh}\left[\log
\frac{\left(\!P_d\hh\!+\!P_d\sdh\!+\!P_a\sdh\!+\!{ \sigma^2}\!\right)^{T_d}(P_d\sdh\!+\!P_a\sdh+\!{ \sigma^2})^{T_d}}
{(P_d\sdh\!+\!P_a\sdh+\!{ \sigma^2})^{T_d}\left(\left(P_a\sdh\!+\!{ \sigma^2}\right)^{T_d}\left|\frac{P_d\sdh \sd\sd^{\dag}}{P_a\sdh\!+\!{ \sigma^2}}\!+\!\bI_{T_d}\right|\right)}\!\right]\\
&=\td \mathbb{E}_{\hat\bh}\left[ \log \left(1+\frac{\pd\hh}{\pd \sdh+P_a\sdh+{ \sigma^2} }\right)\right]+T \dRuGen. \label{upper bound proof without AN 3}
\end{align}

Similarly, it also holds, for $\Delta \bg$ Gaussian and $n_t$ sufficiently large, that
\begin{align}
\notag&
\td \mathbb{E}\left[ \log \left(1+\frac{\pd\gg}{\pd \sdg+ P_a\frac{\ng}{n_t-1}+P_a\sdg+\sigma^2 }\right)\right]\leq I(\sd;\zd|\hat \bh, \hat \bg) \\
\label{upper bound of mutual information of eve without AN}
& \hspace{1cm}\leq \td \mathbb{E}\left[ \log \left(1+\frac{\pd\gg}{\pd \sdg+ P_a\frac{\ng}{n_t-1}+P_a\sdg+\sigma^2 }\right)\right]+T \dRlGen
\end{align}
By combining the above bounds for  $I(\sd;\yd| \hat \bh)$ and $I(\sd;\zd| \hat \bh, \hat \bg)$, we obtain the bounds of the achievable secrecy rate in Theorem \ref{lem upper lower}.

\section{Proof of Theorem \ref{lem small_term} and Corollary \ref{cor linear ratio CONV}}\label{app.conv_case_proof}

\subsection{Proof of Theorem \ref{lem small_term}}

Let us consider the linear power allocation $\mathcal{P}_l\triangleq (P_{f,l}, P_{d,l}, P_{a,l})=(\alpha_f P, \beta_d P, \beta_a P)$ for some positive constants $\alpha_f$, $\beta_d$, and $\beta_a$ such that $\alpha_{f}PT_f+\beta_{d}PT_d+ \beta_{a}PT_d\leq PT$. Then, by Theorem \ref{lem upper lower}, it follows that
\begin{align}\label{bound of optimal rate CONV}
\tRANdata(\calP_l)\!-\!\dRlANdata(\calP_l)\!\leq\! \RANdata(\calP^*)\!\leq\! \tRANdata(\calP^*)\!+\!\dRuANdata(\calP^*).
\end{align}
Hence, to obtain Theorem \ref{lem small_term}, it is sufficient to show that $\tRANdata(\calP_l)-\dRlANdata(\calP_l)\doteq \tRANdata(\calP^*)+\dRuANdata(\calP^*)$, i.e., $\RANdata(\calP^*)\doteq \tRANdata(\calP^*)+\dRuANdata(\calP^*)$, and $\plim \dRuANdata(\calP^*)/\RANdata(\calP^*)=0$.

Specifically, by substituting $\calP_l$ into \eqref{sdh without AN} and \eqref{sdg without AN}, we can express the channel estimation error variances as
$\sdh=\frac{\sh\sigma^2 }{ \alpha_f P \sh+ \sigma^2} = \frac{\sigma^2}{\alpha_f P}+o\left(\frac{1}{P}\right)$ and $\sdg=\frac{\sg \sigma^2 }{\alpha_f P \sg+ \sigma^2 }= \frac{\sigma^2}{\alpha_f P}+o\left(\frac{1}{P}\right)$. Then, it follows that
\begin{align}
\notag  \tRANdata(\calP_l)\!
=&\frac{T_d}{T}\ee\left\{ \log \left[1+\frac{\beta_d P (\sh+o(1))\effh}{(\beta_d+\beta_a) P\left(\frac{\sigma^2}{\alpha_f P}\!+\!o\left(\frac{1}{P}\right)\right)\!+\!\sigma^2 }\right]\right\}\\
&~-\frac{T_d}{T}\ee\!\left\{ \log\! \left[\!1\!+\!\frac{\beta_d P(\sg\!+\!o(1))\effg}{(\beta_d\!+\!\beta_a)P\left(\frac{\sigma^2}{\alpha_f P}\!+\!o\left(\frac{1}{P}\right)\right)\!+\! \beta_a P(\sg+o(1))\frac{\|\bN_{\hat \bh}\mathbf{\overline{g}}\|^2}{n_t-1}\!+\!\sigma^2 }\!\right]\!\right\}\\
=&\frac{T_d}{T}\ee\!\left[\log\frac{\beta_d P\sh\effh\!+\!o(P)}{\frac{ \sigma^2(\beta_d+\beta_a)}{\alpha_f }\!+\!\sigma^2\!+\!o(1) }\!\right]\!-\!\frac{T_d}{T}\ee\left[ \log \!\left(\!1\!+\!\frac{\beta_d P\sg\effg\!+\!o(P)}{\beta_a P\sg\frac{\|\bN_{\hat \bh}\mathbf{\overline{g}}\|^2}{n_t-1}\!+\!o(P) }\!\right)\!\right]\\
=&\frac{T_d}{T}\ee\left[ \log \frac{\beta_d P\sh\effh}{\left(\frac{\beta_d+\beta_a }{\alpha_f }+1\right){ \sigma^2} }\right]
-\frac{T_d}{T}\ee\left[ \log \left(1+\frac{\beta_d \effg}{ \beta_a \frac{\|\bN_{ \hat \bh}\mathbf{\overline{g}}\|^2}{n_t-1} }\right)\right]+o(1)\label{RANdata small o 1}\\
=&\frac{T_d}{T}\log P+ c_1+o(1)\label{RANdata small o 2},
\end{align}
where
\begin{equation}
c_1\triangleq\frac{T_d}{T}\ee\left[ \log \frac{\beta_d \sh\effh}{\left(\frac{\beta_d+\beta_a }{\alpha_f }+1\right){ \sigma^2} }\right]
-\frac{T_d}{T}\ee\left[ \log \left(1+\frac{\beta_d \effg}{ \beta_a \frac{\|\bN_{ \hat \bh}\mathbf{\overline{g}}\|^2}{n_t-1} }\right)\right]
\end{equation}
is a finite constant that is independent of $P$, and
\begin{align}
\notag &\!\!\!\!\dRlANdata(\calP_l)\!\\
\notag=&~\frac{1}{T}\ee\!\left\{\log\frac{\left[(\beta_{d}\!+\!\beta_{a})P\left( \frac{\sigma^2}{\alpha_f P}\!+\!o\left(\frac{1}{P}\right)\right)\!+\!\beta_{a}P\left(\sg\!+\!o(1)\right)\frac{\ngg}{n_t-1}\!+\!\sigma^2\!\right]^{T_d}}
 {\left[\beta_{a}P\left(\sg\!+\!o(1)\right)\frac{\ngg}{n_t-1}\!+\!\beta_{a}P\left( \frac{\sigma^2}{\alpha_f P}\!+\!o\left(\frac{1}{P}\right)\right)\!+\! \sigma^2\!\right]^{T_d-1}}\right\}\\
&-\frac{1}{T}\ee\!\left\{\log\!\left[ (\beta_{d}\|\sd\|^2\!+\!\beta_{a})P\left( \frac{\sigma^2}{\alpha_f P}\!+\!o\left(\frac{1}{P}\right)\right)\!+\!\beta_{a}P\!\left(\sg\!+\!o(1)\right)\!\frac{\ngg}{n_t\!-\!1}\!+\!\sigma^2\!\right]\!\right\}\!\\
=&~\frac{1}{T}\ee\!\left[\log\frac{\left(\beta_{a}P\sg\frac{\ngg}{n_t-1}\!+\!o(P)\!\right)^{T_d}}
 {\left(\beta_{a}P\sg\frac{\ngg}{n_t-1}\!+\!o(P)\!\right)^{\!T_d-1}}\right]\!-\!\frac{1}{T}\ee\!\left[\log\!\left(\!\beta_{a}P\sg\frac{\ngg}{n_t\!-\!1}\!+\!o(P)\!\right)\!\right]
=o(1).\label{Rldelta small o 2}
\end{align}
Hence,
\begin{equation}\label{eq.Rconv_lowerboundproof}
\tRANdata(\calP_l)-\dRlANdata(\calP_l)=\frac{T_d}{T}\log P+c_1+o(1).
\end{equation}

Moreover, by \eqref{R lemma 1} and \eqref{Ru lemma 1}, we can write
\begin{align}
\notag &\tRANdata(\calP^*)+\dRuANdata(\calP^*)\\
\notag&\overset{(a)}{\leq}\frac{T_d}{T}\mathbb{E}\!\left[ \log \! \left(1\!+\!\frac{\pd^*(\sh\!-\!\sdh)\effh}{\pd^* \sdh\!+\!P_a^*\sdh\!+\!\sigma^2 }\!\right)\!\right]\\
&~~~~+\frac{1}{T}\log \left(\pd^*\sdh+P_a^*\sdh+ \sigma^2\right)^{T_d}\!-\!\frac{1}{T}\mathbb{E}\left[\log\!\left(\pd^* \|\sd\|^2\sdh\right)\right]\label{Pf linear ratio 0}\\
&\overset{(b)}{\leq} \frac{T_d}{T}\mathbb{E}\!\left[\log\! \left(k'P\left(2\sdh\!+\!1\!+\!(\sh\!-\!\sdh)\effh\right)\right)\right]
-\frac{1}{T}\mathbb{E}\left[\log\!\left(\pd^* \|\sd\|^2\sdh\right)\right]\\
&= \frac{T_d}{T}\log P-\frac{1}{T}\log\left(\pd^* \sdh \right)+c_2,\label{Pf linear ratio 1}
\end{align}
where $c_2\triangleq(1/T)\ee\!\left[\log\!\left( \left(k'\left(2\sdh\!+\!1\!+\!(\sh\!-\!\sdh)\effh\right)\right)^{T_d}/\|\sd\|^2\right)\right]$ is a finite constant. The inequality in (a) follows by eliminating the negative term of $\tRANdata(\calP^*)$ in \eqref{R lemma 1} and by eliminating some positive parts in the denominator of the first term as well as in the second term of $\dRuANdata$ in \eqref{Ru lemma 1}; and  (b)  is obtained by upper-bounding $P_d$, $P_a$, and $\sigma^2$ by $k'P$, where $k'$ is chosen such that $k'P\geq \max\{P_d, P_a,\sigma^2\}$.  By \eqref{bound of optimal rate CONV}, \eqref{eq.Rconv_lowerboundproof}, and \eqref{Pf linear ratio 1}, it follows that $P_d^*\sigma_{\Delta h}^2=O(1)$ (since otherwise the upper bound in \eqref{Pf linear ratio 1} would be smaller than the lower bound in \eqref{eq.Rconv_lowerboundproof}).
This implies that $\dRuANdata(\calP^*)$ is a finite constant and, thus,
we can write
\begin{align}
\tRANdata(\calP^*)+\dRuANdata(\calP^*)
&\leq \frac{T_d}{T}\mathbb{E}\!\left[ \log \left(1\!+\!\pd^*(\sh\!-\!\sdh)\effh\right)\!\right]+\dRuANdata(\calP^*)\\
&\leq \frac{T_d}{T}\log P +c_2'+o(1),\label{Rconv asy upperbound}
\end{align}
for some constant $c_2'\triangleq\frac{T_d}{T}\ee\left[\log k'(\sh-\sdh+1)
\effh\right]+\dRuANdata(\calP^*)$. By combining \eqref{eq.Rconv_lowerboundproof} and \eqref{Rconv asy upperbound}, we obtain the desired result $\tRANdata(\calP_l)-\dRlANdata(\calP_l)\doteq \tRANdata(\calP^*)+\dRuANdata(\calP^*)$.
Moreover, since $\dRuANdata(\calP^*)$ is finite, we have $\plim \dRuANdata(\calP^*)/ \RANdata(\calP^*)\!=\!0$, which completes the proof.

\subsection{Proof of Corollary \ref{cor linear ratio CONV}}\label{app.proof_cor_2}


The proof of the corollary relies on the fact that
\begin{align}\label{eq.corollary_proof_premise}
\tRANdata(\calP_l)-\dRlANdata(\calP_l) \leq \tRANdata(\calP^*)+\dRuANdata(\calP^*)
\end{align}
for any linear power allocation $\calP_l$.

Specifically, let
us first consider the upper bound
\begin{align}
\RANdata(\calP^*)\overset{(a)}{\leq}
& \frac{T_d}{T}\mathbb{E}\left[ \log  \left(1+\pd^*\sh\effh /\sigma^2\right)\right]+\dRuANdata\\
\overset{(b)}{\leq} & \frac{T_d}{T} \log \left(1+\pd^*\sh/\sigma^2\right)+\dRuANdata,
\end{align}
where (a) is obtained by eliminating the negative terms in $\tRANdata(\calP)$ and by lower-bounding the denominator of the first term by $\sigma^2$, and (b) follows from Jensen's inequality.
By the argument below \eqref{Pf linear ratio 1}, we know that $P_d^*\sigma_{\Delta h}^2=O(1)$, and thus, $R^{(u)}_{\rm conv}(\calP^*)=O(1)$. Then, together with \eqref{eq.corollary_proof_premise} and \eqref{eq.Rconv_lowerboundproof}, we have
$\frac{T_d}{T} \log \left(1+\pd^*\sh/\sigma^2\right)+O(1)\geq \frac{T_d}{T}\log P+c_1+o(1)$,
which implies that $P_d^*(P)=\Omega(P)$. Moreover, since $P_d^*\sigma_{\Delta h}^2=P_d^*\sigma_h^2\sigma^2/(P_f^*\sigma_h^2+\sigma^2)=O(1)$, it also follows that $P_f^*(P)=\Omega(P_d^*(P))=\Omega(P)$.
Furthermore, since $\sigma_{\Delta g}^2/\sigma_{\Delta h}^2=\frac{(P_f^*\sigma_h^2+\sigma^2)/\sigma_h^2}{(P_f^*\sigma_g^2+\sigma^2)/\sigma_g^2}=O(1)$, we know that  $P_d\sigma_{\Delta g}^2=(P_d\sigma_{\Delta h}^2)(\sigma_{\Delta g}^2/\sigma_{\Delta h}^2)=O(1)$. Therefore, the achievable secrecy rate can be upper-bounded as
\begin{align}
\RANdata\!&\leq \tRANdata+\dRuANdata \\
\notag &\leq \frac{T_d}{T}\mathbb{E}\left[ \log \left(1+\pd^*\sh\effh /\sigma^2\right)\right]\\
&~~-\frac{T_d}{T}\mathbb{E}\!\left[ \log\!\left(\!1\!+\!\frac{\pd^*(\sg\!-\!\sdg)\effg}{\pd^* \sdg\!+\! P_a^*(\sg\!-\!\sdg)\frac{\ngg}{n_t\!-\!1}\!+\!P_a^*\sdg\!+\!\sigma^2 }\right)\!\right]\!+\!\dRuANdata\\
&\overset{(a)}{=} \frac{T_d}{T}\mathbb{E}\!\left[ \log\! \left(1\!+\!\pd^*\sh\effh/\sigma^2\right)\right]
\!-\!\frac{T_d}{T}\mathbb{E}\!\left[ \log\! \left(1\!+\!\frac{\pd^*\sg\effg\!+\!o(P_d^*)}{ P_a^*\sg\frac{\ngg}{n_t\!-\!1}\!+\!o(P_d^*)}\right)\right]\!+\!\dRuANdata\\
&= \frac{T_d}{T} \log \pd^*-\frac{T_d}{T}\mathbb{E}\left[ \log \left(1+\frac{\pd^*\effg+o(P^*_d)}{ P^*_a\frac{\ngg}{n_t-1}+o(P^*_d)}\right)\right]+ c_3+o(1),\label{Pa linear 1}
\end{align}
where $c_3=(T_d/T)\ee\left[\log\left(\sh\effh/\sigma^2\right)\right]+\dRuANdata$ is a finite constant and (a) holds since $P^*_d \sdg=O(1)$ and $P^*_a \sdg=O(1)$. From \eqref{Pa linear 1}, we can observe that the second term is a negative finite constant only if $ P^*_a(P)=\Omega(P_d^*(P))$ which implies that $P^*_a(P)=\Omega(P)$.

\section{Proof of \eqref{high_AN_app_line_5} in Section \ref{sec.no AN training}}\label{proof of tRconv lower bound}

By the weak law of large numbers (WLLN), we know that $\ngg/(n_t-1) \rightarrow 1$ in probability as $n_t \rightarrow \infty$. That is, for any $\epsilon >0$, we have $\Pr(A_{\epsilon})\rightarrow 1$ (and, thus, $\Pr(A_{\epsilon}^c)\rightarrow 0$) as $n_t\rightarrow\infty$, where $A_{\epsilon}\triangleq\left\{\left|\frac{\ngg}{n_t-1}-1\right|\leq \epsilon\right\}$.
Therefore, for $n_t$ sufficiently large, the second expectation term in \eqref{RABF approx used in optimization prob} can be written as
\begin{align}
\notag &\!\!\!\!\mathbb{E}\left[\log \left(1+\frac{P^*_d\effg}{P^*_a\frac{\ngg}{(n_t-1)} }\right)\Bigg|A_{\epsilon}\right]\Pr\left(A_{\epsilon}\right)
+\mathbb{E}\left[\log \left(1+\frac{P^*_d\effg}{P^*_a\frac{\ngg}{(n_t-1)} }\right)\Bigg|A_{\epsilon}^c\right]\Pr\left(A_{\epsilon}^c\right)\\
\leq &~ \mathbb{E}\left[\log \left(1+\frac{P^*_d\effg}{P^*_a(1-\epsilon) }\right)\Bigg|A_{\epsilon}\right]\Pr\left(A_{\epsilon}\right)
+\mathbb{E}\left[\log \left(1+\frac{P^*_d\effg}{P^*_a\frac{\ngg}{(n_t-1)} }\right)\Bigg|A_{\epsilon}^c\right]\Pr\left(A_{\epsilon}^c\right)\\
\leq &~\mathbb{E}\left[\log \left(1+\frac{P^*_d\effg}{P^*_a(1-\epsilon) }\right)\right]
+\epsilon_{n_t}.\label{high_AN_app_line_3}
\end{align}
where
\begin{align}\label{epsilon term}
\epsilon_{n_t}\triangleq \mathbb{E}\left[\log \left(1+\frac{P^*_d\effg}{P^*_a\frac{\ngg}{(n_t-1)} }\right)\Bigg|A_{\epsilon}^c\right]\Pr\left(A_{\epsilon}^c\right).
\end{align}
Notice that $\epsilon_{n_t}\rightarrow 0$ as $n_t\rightarrow\infty$ since the expectation inside \eqref{epsilon term} is finite.
Then, by applying Jensen's inequality to the first term in \eqref{high_AN_app_line_3}, we have
\begin{align}\label{high_AN_app_line_4}
\mathbb{E}\left[\log \left(1+\frac{P^*_d\effg}{P^*_a(1-\epsilon) }\right)\right]
\leq \log \left(1+\frac{P^*_d}{P^*_a(1-\epsilon) }\right)
\leq \log \left(\frac{P^*_a+P^*_d}{P^*_a(1-\epsilon) }\right)
\end{align}
Finally, by \eqref{RABF approx used in optimization prob} and \eqref{high_AN_app_line_4}, we have
\begin{align}
\tRANdata\geq\frac{T_d}{T}\log\frac{\frac{P^*_d P^*_a}{P^*_a+P^*_d} }{\left(\frac{P^*_d+P^*_a }{P^*_f }+1 \right)\sigma^2}
+\frac{T_d}{T}\mathbb{E}\left[\log \left(\sh\effh(1-\epsilon) \right)\right]
+\frac{T_d}{T}\epsilon_{n_t}+o(1)
\end{align}

\section{Proof of Theorem \ref{asymptotically equivalent DCE} and Corollary \ref{cor linear ratio DCE}}\label{app.DCE_case_proof}

\subsection{Proof of Theorem \ref{asymptotically equivalent DCE}}

The proof of Theorem \ref{asymptotically equivalent DCE} is an extension of the proof of Theorem \ref{lem small_term} in Appendix \ref{app.conv_case_proof} and, thus, is explained more concisely in the following.

Specifically, let us also consider a linear power allocation $\calP_l\triangleq(P_{r,l},P_{f,l},P_{f_a,l},P_{d,l},P_{a,l})=(\alpha_rP,\alpha_fP,\alpha_{f_a}P,\beta_dP,\beta_aP)$, where $\alpha_r,\alpha_f,\alpha_{f_a},\beta_d$, and $\beta_a$ are positive constants chosen such that the total power constraint in \eqref{eq.power_constraint} is satisfied. Similar to Appendix \ref{app.conv_case_proof},
it is sufficient to show here that $\tRANboth(\calP_l)-\dRlANboth(\calP_l)\doteq \tRANboth(\calP^*)+\dRuANboth(\calP^*)$ and $ \dRuANboth(\calP^*)=O(1)$.

By substituting $\calP_l$ into \eqref{var error h forward training} and \eqref{var error g forward training}, we can write the channel estimation error variances as
\begin{align}\label{repeat sdhDCE}
 \sigma^2_{\Delta h}&=\frac{\sh\left(P_{f_a,l}\frac{\sh\sigma^2}{\sigma^2+P_{r,l}\sh}+\sigma^2\right)}{P_{f_a,l}\frac{\sh\sigma^2}{\sigma^2+P_{r,l}\sh}+\sigma^2+P_{f,l}\sh}
=\frac{(\alpha_{f_a}+\alpha_r)\sigma^2}{\alpha_r\alpha_{f}}P^{-1}+o(P^{-1}),\\
\sigma^2_{\Delta g} &=\frac{P_{f_a,l}\sigma^4_{g}+\sg}{P_{f_a,l}\sg+ \sigma^2+P_{f,l}\sg}
=\frac{\alpha_{f_a}\sg}{\alpha_{f_a}+\alpha_{f}}+o(1).
\end{align}
Thus, we have
\begin{align}
\notag\tRANboth(\calP_l) =&\frac{T_d}{T}\mathbb{E}\left[ \log \left(1+\frac{\beta_dP(\sh+o(1))\effh}{\left(\beta_d+\beta_a\right) \frac{(\alpha_{f_a}+\alpha_r) \sigma^2}{\alpha_r\alpha_{f}}+o(1)+ \sigma^2 }\right)\right]\\
&-\frac{T_d}{T}\mathbb{E}\!\left[ \!\log\! \left(\!1\!+\!\frac{\beta_dP\left(\frac{\alpha_{f}\sg}{\alpha_{f_a}+\alpha_{f}}+o(1)\right)\effg}{ (\beta_d+\beta_a)P\!\left(\frac{\alpha_{f_a}\sg}{\alpha_{f_a}+\alpha_{f}}\!+\!o(1)\!\right)\!+\! \beta_aP\!\left(\frac{\alpha_{f}\sg}{\alpha_{f_a}+\alpha_{f}}\!+\!o(1)\!\right)\!\frac{\ngg}{n_t-1}\!+\! \sigma^2 }\!\right)\!\right]\\
 \notag=&\frac{T_d}{T}\mathbb{E}\left[ \log \left(\frac{\beta_dP\sh\effh}{\left(\beta_d+\beta_a\right) \frac{(\alpha_{f_a}+\alpha_r)\sigma^2}{\alpha_r\alpha_{f}}+\sigma^2 }\right)\right]\\
 &-\frac{T_d}{T}\mathbb{E}\left[\log \left(1+\frac{\beta_d\alpha_{f}\effg}{ (\beta_d+\beta_a)\alpha_{f_a}+ \beta_a\alpha_{f}\frac{\ngg}{n_t-1} }\right)\right]+o(1)\label{repeat tRANboth with linear ratio}\\
=&\frac{T_d}{T}\log P +  c_4+o(1).
\end{align}
where $c_4\triangleq \frac{T_d}{T}\mathbb{E}\left[ \log \left(\frac{\beta_d\sh\effh}{\left(\beta_d+\beta_a\right) \frac{(\alpha_{f_a}+\alpha_r)\sigma^2}{\alpha_r\alpha_{f}}+ \sigma^2 }\right)\right]-\frac{T_d}{T}\mathbb{E}\left[\log \left(1+\frac{\beta_d\alpha_{f}\effg}{ (\beta_d+\beta_a)\alpha_{f_a}+ \beta_a\alpha_{f}\frac{\ngg}{n_t-1} }\right)\right]$. Moreover, by substituting $\calP_l$ into \eqref{Rl lemma 1}, it can also be shown that $ \dRlANboth(\calP_l)=O(1)$. Hence, we have $\tRANboth(\calP_l)-\dRlANboth(\calP_l)=\frac{T_d}{T}\log P + O(1).$

Furthermore, by following the derivations in \eqref{Pf linear ratio 0}-\eqref{Pf linear ratio 1}, it can also be shown that $P_d^*\sigma_{\Delta h}^2=O(1)$ and, thus,
$\tRANboth(\calP^*)+\dRuANboth(\calP^*) \leq \frac{T_d}{T}\log P+O(1).$
Hence, it follows that $\tRANboth(\calP_l)-\dRlANboth(\calP_l)\doteq \tRANboth(\calP^*)+\dRuANboth(\calP^*)$.

\subsection{Proof of Corollary \ref{cor linear ratio DCE}}

The proofs of  $P^*_f(P)=\Omega(P)$ and $P^*_d(P)=\Omega(P)$ are the same as those in Appendix \ref{app.conv_case_proof} for the case with conventional training. Hence, we prove here that either $P_{fa}^*(P)=\Omega(P)$ or $P_a^*(P)=\Omega(P)$, and that $P_r^*(P)=\Omega(P_{fa}^*(P))$.

First, let us recall that
\begin{align}\label{pd times sdh finite DCE}
P_d^*\sdh=\frac{ P_d^*\sh \left(P_{f_a}^*\sdhr+ \sigma^2\right)}{P_f^*\sh+P_{f_a}^*\sdhr+ \sigma^2}=O(1).
\end{align}
Due to the total power constraint, it holds that $P_f^*\sh+P_{f_a}^*\sdhr+\sigma^2=O(P)$. Therefore, together with the fact that $P^*_d(P)=\Omega(P)$, it follows that $P_{f_a}^*\sdhr=O(1)$. Then, by \eqref{var error h reverse training}, we have $P_{f_a}^*\sdhr=\frac{P_{f_a}^*\sh \sigma^2}{P_r^*\sh+ \sigma^2}=O(1)$, which implies that $P^*_r(P)=\Omega(P^*_{f_a}(P))$.


Finally, to show that
$P^*_{f_a}(P)=\Omega(P)$ or $P^*_{a}(P)=\Omega(P)$, let us rewrite the upper bound of the achievable secrecy rate as follow:
\begin{align}
&\tRANboth+\dRuANboth \\
\notag &\leq \frac{T_d}{T}\mathbb{E}\left[ \log \left(1+\pd^*\sh\effh\right)\right]\\
&~~~~-\frac{T_d}{T}\mathbb{E}\left[ \log \left(1+\frac{\pd^*(\sg-\sdg)\effg}{\pd^* \sdg+ P_a^*(\sg-\sdg)\frac{\ngg}{n_t-1}+P_a^*\sdg+\sigma^2 }\right)\right]+\dRuANboth\\
& =\frac{T_d}{T}\log \pd^* -\frac{T_d}{T}\mathbb{E}\left[ \log \left(1+\frac{\pd^*(\sg-\sdg)\effg}{\pd^* \sdg\!+\! P_a^*(\sg\!-\!\sdg)\frac{\ngg}{n_t-1}\!+\!P_a^*\sdg\!+\!\sigma^2 }\right)\right]\!+\!O(1).
\label{either Pa or Pfa proof}
\end{align}
The last equality comes from the fact that  $\dRuANboth(\calP^*)=O(1)$ since $ P_d^*\sdh=O(1)$. Then, by the fact that $\tilde R_{\rm DCE}(\calP^*)+\Delta R_{\rm DCE}^{(u)}(\calP^*)\geq \tilde R_{\rm DCE}(\calP_l)-\Delta R_{\rm DCE}^{(l)}(\calP_l)=\frac{T_d}{T}\log P + O(1)$, it follows that the second term in \eqref{either Pa or Pfa proof} must be $O(1)$. This implies that term inside the logarithm must be $O(1)$.  By substituting $\sigma_{\Delta g}^2$ with \eqref{var error g forward training} and using $T_f=n_t$, this term can be written more explicitly as
\begin{align*}
&\frac{\pd^*(\sg-\sdg)\effg}{\pd^* \sdg+ P_a^*(\sg-\sdg)\frac{\ngg}{n_t-1}+P_a^*\sdg+ \sigma^2 }\\
&=\frac{\pd^*P_f^*\sigma_g^4 \effg}{(\pd^*+P_a^*) \sg(P_{fa}^*\sg+\sigma^2)+ P_a^*P_f^*\sigma_g^4 \frac{\ngg}{n_t-1}+\sigma^2(P_{fa}^*\sg+\sigma^2+\sg P_f^* ) }.
\end{align*}
Since $P^*_f(P)=\Omega(P)$ and $P^*_d(P)=\Omega(P)$, it is necessary to have either $P_a^*(P)=\Omega(P)$ or $P_{fa}^*(P)=\Omega(P)$ (or both) in order for this term to scale as $O(1)$. This completes the proof.

\section{Proof of Corollary \ref{lem sufficient condition} in Section \ref{Sec.DCE}}\label{sec. proof sufficient condition DCE is better}

To distinguish between the conventional and the DCE cases, let us denote the power allocation in the conventional case as $\calP=(P_f, P_d, P_a)$ and that in the DCE case as $\calQ=(Q_r, Q_{f_a}, Q_f, Q_d, Q_a)$. The approximate achievable secrecy rate $\tilde R_{\rm conv}(\hat\calP^*)$ is given by \eqref{RABF approx used in optimization prob} and the optimal power allocation $\hat\calP^*$ in the conventional case is given in \eqref{optimal power ratio RConv}. Corollary \ref{lem sufficient condition} is proved by showing that a lower bound of $\tilde R_{\rm DCE}(\calQ^*)$, where $\calQ^*$ is the optimal power allocation in the DCE case, is greater than an upper bound of $\tilde R_{\rm conv}(\hat\calP^*)$ if the condition in \eqref{sufficient condition of T} is satisfied.

To obtain an upper bound for $\tilde R_{\rm conv}(\hat\calP^*)$, let us first note,
by WLLN, that $\|\bN_{\hat\bh}\bar\bg\|^2/(n_t-1)\rightarrow 1$ and $\bar\bg^\dag\bar\bh\bar\bh^\dag\bar\bg/\|\bar \bh\|^2\rightarrow 1$ in probability as $n_t\rightarrow 1$. Therefore, by defining $A_{\epsilon}'\triangleq \{|\frac{\|\bN_{\hat\bh}\bar\bg\|^2}{(n_t-1)}-1|\leq \epsilon, |\effg-1|\leq \epsilon\}$, the expectation inside second term of \eqref{RABF approx used in optimization prob} can be lower bounded by
$\mathbb{E}\left[\log \left(1+\frac{\hat P_d^*(1-\epsilon)}{\hat P_a^*(1+\epsilon) }\right)\Bigg|A_{\epsilon}'\right]\Pr\left(A_{\epsilon}'\right)
+\mathbb{E}\left[\log \left(1+\frac{\hat P_d^*\effg}{\hat P_a^*\frac{\ngg}{(n_t-1)} }\right)\Bigg|{A_{\epsilon}'}^c\right]\Pr\left({A_{\epsilon}'}^c\right)
=\log \left(1+\frac{\hat P_d^*(1-\epsilon)}{\hat P_a^*(1+\epsilon) }\right)
+\epsilon_{n_t}''$,
where $\epsilon_{n_t}''\rightarrow 0$ as $n_t\rightarrow \infty$. By substituting this 
into  \eqref{RABF approx used in optimization prob}, we have
\begin{align}\label{suff condtion ANdata derivation 1}
\tRANdata(\hat \calP^*)&\leq
\frac{T_d}{T}\ee\left[ \log \frac{\hat P_d^*\sh\effh}{\left(\frac{\hat P_d^*+\hat P_a^* }{\hat P_f^* }+1\right){\sigma^2} }\right]
-\frac{T_d}{T}\log \left(1+\frac{\hat P_d^*(1-\epsilon)}{\hat P_a^*(1+\epsilon) }\right)
-\frac{T_d}{T}\epsilon_{n_t}''+o(1)\\\label{suff condtion ANdata derivation 2}
&=
\frac{T-n_t}{T}\left\{\log\frac{\hat P^*_d}
{2\left(\frac{2\hat P_d^*}{\hat P^*_f}+1\right)}+\ee\left[\log\frac{\sh\effh(1+\epsilon)}
{\sigma^2}\right]+\epsilon_{n_t}''\right\}+o(1).
\end{align}
since $T_d=T-T_f=T-n_t$ in the conventional training based scheme and $\hat P_d^*=\hat P_a^*$ (c.f. \eqref{optimal power ratio RConv}).

To obtain a lower bound for $\tilde R_{\rm DCE}(\calQ^*)$, we consider the power allocation policy $\calQ^\sharp$ defined such that $Q_r^\sharp=\frac{\gamma T_f}{2T_r}\hat P_f^*=\frac{\gamma n_t}{2}\hat P_f^*$, $Q_f^\sharp=(1-\gamma)\hat P_f^*$, $Q_{f_a}^\sharp=\frac{\gamma}{2}\hat P_f^*$, $Q_d^\sharp=\hat P_d^*$, and $Q_a^\sharp=\hat P_a^*$, where $\gamma$ is a constant in $(0,1)$. Since $\calQ^\sharp$ is an arbitrarily chosen power allocation policy, it follows that $\tilde R_{\rm DCE}(\calQ^*)\geq \tilde R_{\rm DCE}(\calQ^\sharp)$. Notice that $\calQ^\sharp$ is similar to $\hat \calP^*$, but with $\gamma$ portion of the training energy moved to the reverse pilot signal and to AN in the training phase. It is also worthwhile to note that, even though the power allocated to signal and AN in the data transmission phase, i.e., $Q_d^\sharp=\hat P_d^*$, and $Q_a^\sharp=\hat P_a^*$, are the same, the total energy expended in the data transmission phase is smaller than that in the conventional training based scheme since $T_d=T-n_t-1$ in this case (i.e., an additional channel use is spent for reverse training). Hence, the total energy consumed by $\calQ^\sharp$ is actually strictly less than the constraint $PT$.

By the fact that all power components in $\calQ^\sharp$ scale linearly with $P$ as in $\hat \calP^*$ and by substituting $\calQ^\sharp$ into \eqref{RDCE approximation final ver}, we have
\begin{align}
\notag\tRANboth(\calQ^\sharp)\!&\!\geq\!
\frac{T\!-\!n_t-\!1\!}{T}\left\{\log \frac{\hat P^*_d}{\frac{2\hat P^*_d}{\hat P_f^*}\frac{1+n_t}{n_t(1-\gamma)}\!+\!1 }\!+\ee\left[ \log\frac{\sh\effh}{\sigma^2}\right]\!-\log \frac{2\!-\!\gamma\!-\!(1\!-\!\gamma)\epsilon'}{1\!-\!\epsilon'\!+\!\epsilon'\gamma}\!+\!\epsilon'_{n_t}\right\}\!+\!o(1)\\
&\geq\! \frac{T\!-\!n_t-\!1\!}{T}\left\{ \log\frac{\hat P^*_d}{\frac{2\hat P^*_d}{\hat P_f^*}\frac{1+n_t}{n_t(1-\gamma)}\!+\!1 }\!+\ee\left[ \log\frac{\sh\effh}{\sigma^2}\right]- \log \frac{2\!-\!\gamma}{1\!-\!\epsilon'}\!+\!\epsilon'_{n_t}\right\}\!+\!o(1).\label{suff condtion ANBoth derivation 2}
\end{align}

By \eqref{suff condtion ANdata derivation 2} and \eqref{suff condtion ANBoth derivation 2}, the difference between $\tRANboth(\calQ^*)$ and $\tRANdata(\hat \calP^*)$ can be lower bounded as
\begin{align}
\notag&\tRANboth(\calQ^*)\!-\!\tRANdata(\hat \calP^*)\\
 \notag &\geq \frac{T\!-\!n_t-\!1\!}{T}\left[\log \frac{2\left(\frac{2\hat P_d^*}{\hat P_f^*}+1\right)}{\frac{2\hat P^*_d}{\hat P_f^*}\frac{1+n_t}{n_t(1-\gamma)}+1}-\log \frac{ 1+\epsilon}{1-\epsilon'}-\log(2-\gamma)+\epsilon'_{n_t}\right]\\
&~~~~- \frac{1}{T} \ee\left[\log \frac{\hat P_d^* \sh\effh( 1+\epsilon)/\sigma^2}{2\left(\frac{2\hat P^*_d}{\hat P_f^*} +1\right) }\right]-\frac{T-n_t}{T}\epsilon_{n_t}''+o(1)\\
\notag &= \frac{T\!-\!n_t-\!1\!}{T}\left[\log \frac{2\left(\sqrt{\frac{n_t}{T-n_t}}+1\right)}{
\sqrt{\frac{n_t}{T-n_t}}\frac{1+n_t}{n_t(1-\gamma)}+1}-\log \frac{ 1+\epsilon}{1-\epsilon'}-\log(2-\gamma)+\epsilon'_{n_t}\right]\\
&~~~~- \frac{1}{T} \ee\left[\log \frac{PT \sh\effh( 1+\epsilon)/\sigma^2}{4\left(\sqrt{n_t}+\sqrt{T-n_t}\right)^2 }\right]-\frac{T-n_t}{T}\epsilon_{n_t}''+o(1)\label{suff condtion ANBoth derivation 3}
\end{align}

The expectation inside the second term of \eqref{suff condtion ANBoth derivation 3} can be upper bounded as
\begin{align}
 \ee\left[\log \frac{PT \sh\effh(1\!+\!\epsilon)/\sigma^2}{4\left(\sqrt{n_t}+\sqrt{T-n_t}\right)^2 }\right]\!\leq\! \ee\left[\log \left(\frac{P T\sh\effh( 1\!+\!\epsilon)}{4T\sigma^2 }\right)\right]\!\leq\! \log \left(\frac{P\sh n_t( 1\!+\!\epsilon)}{4\sigma^2 }\right),
\end{align}
where the last inequality follows from Jensen's inequality. Therefore, the difference $\tRANboth(\calQ^*)-\tRANdata(\hat \calP^*)$ can be further bounded as
\begin{align}
\notag &\!\!\!\!\!\!\tRANboth(\calQ^*)-\tRANdata(\hat \calP^*) \\
\notag\geq &  \frac{T\!-\!n_t\!-\!1}{T}\log \frac{2\left(\sqrt{\frac{n_t}{T-n_t}}+1\right)}
{\left(\sqrt{\frac{n_t}{T-n_t}}\frac{1+n_t}{n_t(1-\gamma)}+1\right)(2-\gamma)}- \frac{1}{T} \log { \frac{P \sh n_t(1\!+\!\epsilon)}{4\sigma^2 }}\\
&-\frac{T\!-\!n_t\!-\!1}{T}\log\frac{1+\epsilon}{1-\epsilon'}
+\frac{T\!-\!n_t\!-\!1}{T}\epsilon'_{n_t}-\frac{T-n_t}{T}\epsilon''_{n_t}+o(1)
\end{align}
When $P$ and $n_t$ are sufficiently large, we can choose arbitrary small $\epsilon, \epsilon'>0$ such that $\epsilon_{n_t}$, $\epsilon_{n_t}'$, and $o(1)$ can be neglected. Hence, for $\tRANboth(\calQ^*)-\tRANdata(\hat \calP^*)>0$, it is sufficient to have
\begin{align}\label{suff condtion ANBoth derivation 4}
(T-n_t-1)\log\left(\frac{2\left(n_t+\sqrt{(T-n_t)n_t}\right)}{1+n_t+(1-\gamma)\sqrt{(T-n_t)n_t}}\cdot \frac{1-\gamma}{2-\gamma}\right)> \log\left( \frac{P\sh n_t}{4\sigma^2}\right)
\end{align}
By selecting $\gamma = 1/2$,
the term inside the logarithm on the left-hand-side
can be rewritten as
\begin{align}\label{suff condtion ANBoth derivation 5}
\frac{4n_t\!+\!4\sqrt{(T\!-\!n_t)n_t}}{6(1\!+\!n_t)\!+\!3\sqrt{(T\!-\!n_t)n_t}}
=\frac{(12n_t\!+\!2\sqrt{(T\!-\!n_t)n_t})\!+\!10\sqrt{(T\!-\!n_t)n_t}}{18(1\!+\!n_t)\!+\!9\sqrt{(T\!-\!n_t)n_t}}.
\end{align}
Notice that, if $12n_t+2\sqrt{(T-n_t)n_t}\geq 20(1+n_t)$, i.e., if
\begin{equation}\label{eq.app_suffcond1}
T\geq \frac{(4n_t+10)^2}{n_t}+n_t,
\end{equation}
the left-hand-side of \eqref{suff condtion ANBoth derivation 4} is lower bounded by $(T-n_t-1)\log\left(10/9 \right)$. Therefore, we have $\tRANboth(\calQ^*)-\tRANdata(\calP^*)>0$ if
$(T-1-n_t)\log_{10}\left(10/9 \right)>\log_{10}\left(P\sh n_t/4\sigma^2\right)$, which yields the sufficient condition
\begin{align}\label{eq.app_suffcond2}
T\geq 22\log_{10}\left(\frac{P\sh n_t}{4\sigma^2}\right)+1+n_t
\end{align}
since $(\log_{10}(10/9))^{-1}\leq 22$. By \eqref{eq.app_suffcond1} and \eqref{eq.app_suffcond2}, we obtain the result in \eqref{sufficient condition of T}.

\end{appendices}


\begin{thebibliography}{10}
\providecommand{\url}[1]{#1}
\csname url@samestyle\endcsname
\providecommand{\newblock}{\relax}
\providecommand{\bibinfo}[2]{#2}
\providecommand{\BIBentrySTDinterwordspacing}{\spaceskip=0pt\relax}
\providecommand{\BIBentryALTinterwordstretchfactor}{4}
\providecommand{\BIBentryALTinterwordspacing}{\spaceskip=\fontdimen2\font plus
\BIBentryALTinterwordstretchfactor\fontdimen3\font minus
  \fontdimen4\font\relax}
\providecommand{\BIBforeignlanguage}[2]{{%
\expandafter\ifx\csname l@#1\endcsname\relax
\typeout{** WARNING: IEEEtran.bst: No hyphenation pattern has been}%
\typeout{** loaded for the language `#1'. Using the pattern for}%
\typeout{** the default language instead.}%
\else
\language=\csname l@#1\endcsname
\fi
#2}}
\providecommand{\BIBdecl}{\relax}
\BIBdecl

\bibitem{ADWyner1975}
A.~D. Wyner, ``{The wire-tap channel},'' \emph{Bell System Technical Journal},
  vol.~54, pp. 1355--1387, Oct. 1975.

\bibitem{ICsiszar1978}
I.~Csisz\'ar and J.~K\"orner, ``Broadcast channels with confidential
  messages,'' \emph{{IEEE} Trans. Inf. Theory}, vol.~24, pp. 339--348, May
  1978.

\bibitem{SLeung-Yan-Cheong1978}
S.~Leung-Yan-Cheong and M.~Hellman, ``The {Gaussian} wiretap channel,''
  \emph{{IEEE} Trans. Inf. Theory}, vol.~24, no.~4, pp. 451--456, Jul. 1978.

\bibitem{AKhisti2010_MISOME1}
A.~Khisti and G.~Wornell, ``{Secure transmission with multiple antennas
  \textrm{I}: The MISOME wiretap channel},'' \emph{{IEEE} Trans. Inf. Theory},
  vol.~56, no.~7, pp. 3088--3104, Jul. 2010.

\bibitem{AKhisti2010_MIMOME2}
------, ``{Secure transmission with multiple antennas \textrm{II}: The MIMOME
  wiretap channel},'' \emph{{IEEE} Trans. Inf. Theory}, vol.~56, pp.
  5515--5532, Nov. 2010.

\bibitem{SShafiee_Ulukus2009_221}
S.~Shafiee, N.~Liu, and S.~Ulukus, ``Towards the secrecy capacity of the
  {G}aussian {MIMO} wire-tap channel: The 2-2-1 channel,'' \emph{{IEEE} Trans.
  Inf. Theory}, vol.~55, no.~9, pp. 4033--4039, Sep. 2009.

\bibitem{FOggier_Hassibi2011}
F.~Oggier and B.~Hassibi, ``The secrecy capacity of the {MIMO} wiretap
  channel,'' \emph{{IEEE} Trans. Inf. Theory}, vol.~57, no.~8, pp. 4961--4972,
  Aug. 2011.

\bibitem{TLiu_Shamai2009}
T.~Liu and S.~Shamai, ``A note on the secrecy capacity of the multi-antenna
  wiretap channel,'' \emph{{IEEE} Trans. Inf. Theory}, vol.~55, no.~6, pp.
  2547--2553, Jun. 2009.

\bibitem{MBiguesh_Gershman2006}
M.~Biguesh and A.~Gershman, ``Training-based {MIMO} channel estimation: a study
  of estimator tradeoffs and optimal training signals,'' \emph{{IEEE} Trans.
  Signal Process.}, vol.~54, no.~3, pp. 884--893, Aug. 2006.

\bibitem{TFWong_Park2004}
T.~F. Wong and B.~Park, ``Training sequence optimization in mimo system's with
  colored interference,'' \emph{{IEEE} Trans. Commun.}, vol.~52, no.~11, pp.
  1939--1947, Nov. 2004.

\bibitem{JWChen_Ng2008}
J.~W. Chen, Y.~C. Wu, S.~D. Ma, and T.~S. Ng, ``Joint {CFO} and channel
  estimation for multiuser {MIMO-OFDM} systems with optimal training
  sequences,'' \emph{{IEEE} Trans. Signal Process.}, vol.~56, no.~8, pp.
  4008--4019, Aug. 2008.

\bibitem{Chang_Chi2010_DCE}
T.-H. Chang, W.-C. Chiang, Y.-W. Hong, and C.-Y. Chi, ``Training sequence
  design for discriminatory channel estimation in wireless {MIMO} systems,''
  \emph{{IEEE} Trans. Signal Process.}, vol.~58, no.~12, pp. 6223--6237, Dec.
  2010.

\bibitem{CWHuang_Hong2013_TwoWay}
C.-W. Huang, X.~Zhou, T.-H. Chang, and Y.-W. Hong, ``{Two-Way} training for
  discriminatory channel estimation in wireless {MIMO} systems,'' \emph{{IEEE}
  Trans. Signal Process.}, vol.~61, no.~10, pp. 2724--2738, May 2013.

\bibitem{BHassibi_Hochwald2003_Training}
B.~Hassibi and B.~Hochwald, ``How much training is needed in multiple-antenna
  wireless links?'' \emph{{IEEE} Trans. Inf. Theory}, vol.~49, pp. 951--963,
  Apr. 2003.

\bibitem{GCaire2010}
G.~Caire, N.~Jindal, M.~Kobayashi, and N.~Ravindran, ``{Multiuser MIMO
  achievable rates with downlink training and channel state feedback},''
  \emph{{IEEE} Trans. Inf. Theory}, vol.~56, no.~6, pp. 2845--2866, Jun. 2010.

\bibitem{ZRezki_Khisti2014}
Z.~Rezki, A.~Khisti, and M.-S. Alouini, ``On the secrecy capacity of the
  wiretap channel with imperfect main channel estimation,'' \emph{{IEEE} Trans.
  Commun.}, vol.~62, no.~10, pp. 3652--3664, Oct. 2014.

\bibitem{Zhou_McKay_2010}
X.~Zhou and M.~McKay, ``Secure transmission with artificial noise over fading
  channels: Achievable rate and optimal power allocation,'' \emph{{IEEE} Trans.
  Veh. Technol.}, vol.~59, no.~8, pp. 3831--3842, Oct. 2010.

\bibitem{SCLin_Chi2011_Quantized_Feedback}
S.-C. Lin, T.-H. Chang, Y.-L. Liang, Y.-W.~P. Hong, and C.-Y. Chi, ``On the
  impact of quantized channel feedback in guaranteeing secrecy with artificial
  noise: The noise leakage problem,'' \emph{{IEEE} Trans. Wireless Commun.},
  vol.~10, no.~3, pp. 901--915, Mar. 2011.

\bibitem{TYLiu_ICC}
T.-Y. Liu, S.-C. Lin, T.-H. Chang, and Y.-W.~P. Hong, ``How much training is
  enough for secrecy beamforming with artificial noise,'' in \emph{Proc. of
  IEEE International Conference on Communications (ICC)}, Jun. 2012, pp.
  4782--4787.

\bibitem{SGoel_Negi2008}
S.~Goel and R.~Negi, ``Guaranteeing secrecy using artificial noise,''
  \emph{{IEEE} Trans. Wireless Commun.}, vol.~7, no.~6, pp. 2180--2189, Jun.
  2008.

\bibitem{PHLin_SCLin2013}
P.-H. Lin, S.-H. Lai, S.-C. Lin, and H.-J. Su, ``On secrecy rate of the
  generalized artificial noise assisted secure beamforming for wiretap
  channel,'' \emph{{IEEE} J. Sel. Areas Commun.}, vol.~31, no.~9, pp.
  1728--1740, Sep. 2013.

\bibitem{TYLiu_Mag}
T.-Y. Liu, P.-H. Lin, S.-C. Lin, Y.-W.~P. Hong, and E.~A. Jorswieck, ``To avoid
  or not to avoid {CSI} leakage in physical layer secret communication
  systems,'' \emph{{IEEE} Commun. Mag.}, Dec 2015, to appear.

\bibitem{XHe_Yener2014}
X.~He and A.~Yener, ``{MIMO} wiretap channels with unknown and varying
  eavesdropper channel states,'' \emph{{IEEE} Trans. Inf. Theory}, vol.~60,
  no.~11, pp. 6844--6869, Nov. 2014.

\bibitem{TYoo_Goldsmith}
T.~Yoo and A.~Goldsmith, ``Capacity and power allocation for fading {MIMO}
  channels with channel estimation error,'' \emph{{IEEE} Trans. Inf. Theory},
  vol.~52, no.~5, pp. 2203--2214, May 2006.

\bibitem{EstimationTheory}
S.~M. Kay, \emph{Fundamentals of Statistical Signal Processing: Estimation
  Theory}, 1st~ed.\hskip 1em plus 0.5em minus 0.4em\relax Prentice Hall, 1993.

\bibitem{IntroToAlgorithms}
T.~H. Cormen, C.~E. Leiserson, R.~L. Rivest, and C.~Stein, \emph{Introduction
  to Algorithms}, 3rd~ed.\hskip 1em plus 0.5em minus 0.4em\relax MIT Press,
  2009.

\bibitem{MatrixAnalysis}
R.~A. Horn and C.~R. Johnson, \emph{Matrix Analysis}.\hskip 1em plus 0.5em
  minus 0.4em\relax Cambridge University Press, 1990.

\bibitem{MChiang_GeometricProgramming}
M.~Chiang, C.~W. Tan, D.~P. Palomar, D.~O'Neill, and D.~Julian, ``Power control
  by geometric programming,'' \emph{{IEEE} Trans. Wireless Commun.}, vol.~6,
  no.~7, pp. 2640--2651, Jul. 2007.

\bibitem{VVPrelov_Verdu2004_LowSNR}
V.~V. Prelov and S.~Verdu, ``Second-order asymptotics of mutual information,''
  \emph{{IEEE} Trans. Inf. Theory}, vol.~50, no.~8, pp. 1567--1580, Aug. 2004.

\bibitem{CRao_Hassibi2004_LowSNR}
C.~Rao and B.~Hassibi, ``Analysis of multiple-antenna wireless links at low
  {SNR},'' \emph{{IEEE} Trans. Inf. Theory}, vol.~50, no.~9, pp. 2123--2130,
  Sep. 2004.

\end{thebibliography}

\end{document}